\documentclass[a4paper]{article}
\usepackage[margin=2.7cm]{geometry}
\usepackage[utf8]{inputenc}
\usepackage[english]{babel}
\usepackage{amsmath,amsfonts,amssymb}
\usepackage{amsthm}
\usepackage{bbold}
\usepackage{hyperref}
\usepackage{textcomp}
\usepackage{mathrsfs}
\usepackage{graphicx,color}
\usepackage{todonotes}
\usepackage{mathtools}
\usepackage{caption,subcaption}
\usepackage{tikz}
\usepackage{pgfplots}
\usepackage{cite}

\usepackage{pgfkeys}
    \newenvironment{customlegend}[1][]{
        \begingroup
        \csname pgfplots@init@cleared@structures\endcsname
        \pgfplotsset{#1}
    }{
        \csname pgfplots@createlegend\endcsname
        \endgroup
    }
    \def\addlegendimage{\csname pgfplots@addlegendimage\endcsname}

\newtheorem{theorem}{Theorem}
\newtheorem{proposition}[theorem]{Proposition}
\newtheorem{lemma}[theorem]{Lemma}
\newtheorem{corollary}[theorem]{Corollary}
\newtheorem{remark}[theorem]{Remark}

\numberwithin{equation}{section}

\newcommand{\overhat}[1]{\expandafter\hat#1}
\newcommand{\pur}{\mathtt{P}}

\newcommand{\Fnorm}[1]{\|{#1}\|_2}
\newcommand{\infnorm}[1]{\|{#1}\|_\infty}

\newcommand{\pnorm}[2]{\|{#1}\|_{#2}}

\newcommand{\vecket}[1]{|{#1})}
\newcommand{\vecbra}[1]{({#1}|}

\newcommand{\transpose}{^\intercal}
\newcommand{\order}[1]{\mathcal{O}({#1})}

\newcommand{\E}{\operatorname{\mathbb{E}}}
\newcommand{\Var}{\operatorname{Var}}
\renewcommand{\O}{\mathcal{O}}

\newcommand{\rmi}{\mathrm{i}}
\newcommand{\rme}{\mathrm{e}}
\newcommand{\rmd}{\mathrm{d}}
\newcommand{\tr}{\operatorname{tr}}
\newcommand{\ket}[1]{|{#1}\rangle}
\newcommand{\bra}[1]{\langle{#1}|}

\definecolor{dgreen}{rgb}{0,0.5,0}

\definecolor{delete}{cmyk}{0.5,0,0,0}

\makeatletter
\def\@makefnmark{
	\leavevmode
	\raise.9ex\hbox{\fontsize\sf@size\z@\normalfont\tiny\@thefnmark}}
\makeatother

\begin{document}
\title{Unification of Random Dynamical Decoupling and the Quantum Zeno Effect}
\author{Alexander Hahn\thanks{Center for Engineered Quantum Systems, Macquarie University, 2109 NSW, Australia} \textsuperscript{,}\footnote{Institut f\"ur Theoretische Physik, Leibniz Universit\"at Hannover, Appelstra{\ss}e 2, D-30167 Hannover, Germany}\and Daniel Burgarth\footnotemark[1]\and Kazuya Yuasa\thanks{Department of Physics, Waseda University, Tokyo 169-8555, Japan}}

\maketitle

\begin{abstract}
Periodic deterministic bang-bang dynamical decoupling and the quantum Zeno effect are known to emerge from the same physical mechanism. Both concepts are based on cycles of strong and frequent kicks provoking a subdivision of the Hilbert space into independent subspaces. However, previous unification results do not capture the case of random bang-bang dynamical decoupling, which can be advantageous to the deterministic case but has an inherently acyclic structure. Here, we establish a correspondence between random dynamical decoupling and the quantum Zeno effect by investigating the average over random decoupling evolutions. This protocol is a manifestation of the quantum Zeno dynamics and leads to a unitary bath evolution. By providing a framework that we call \emph{equitability of system and bath}, we show that the system dynamics under random dynamical decoupling converges to a unitary with a decoupling error that characteristically depends on the convergence speed of the Zeno limit. This reveals a unification of the random dynamical decoupling and the quantum Zeno effect.
\end{abstract}

\section{Introduction}\label{sec:Introduction}
The inherent interaction of a quantum system with a surrounding environment causes loss of its initial coherences over time~\cite{Lidar2013}. In particular, in the contemporary NISQ era~\cite{Preskill2018}, such environmental noise is a major obstacle in engineering quantum devices. A prosperous technique to overcome this hurdle is to remove undesired couplings by applying strong controls to the system. This procedure falls under the name of \emph{dynamical decoupling} (see e.g.~Refs.~\cite{Viola1998, Ban1998, Viola2002, Cappellaro2006, Uhrig2007}) and historically originates from the field of NMR~\cite{Hahn1950, Waugh1968, Haeberlen1968}. Its intuitive idea is to perform cycles of frequent and strong unitary kicks, e.g.~laser pulses, on the system with a high repetition rate. If the time intervals between the kicks are shorter than the system-bath interaction time, transitions among the system and environmental subspaces are suppressed due to a rotational average effect~\cite{Viola1999}. In turn, one obtains a decoupled evolution generated by an effective projected (Zeno) Hamiltonian~\cite{Facchi2004, Facchi2005, Burgarth2018, ref:OneBoundRWA}. Likewise, the very same idea underlies the \emph{quantum Zeno effect}~\cite{Misra1977, Home1997, Facchi2001a, Facchi2001, ref:PaoloSaverio-QZEreview-JPA}, where projective measurements are performed on the system frequently. This leads to an effective Zeno dynamics, which takes place independently in the subspaces specified by the applied projections~\cite{FP2002}.

The apparent connection between the dynamical decoupling and the quantum Zeno effect suggests the existence of a framework unifying both concepts. Indeed, such a consolidation has been initiated in Ref.~\cite{Facchi2004} via von Neumann's ergodic theorem~\cite[Theorem~II.11]{Reed1980}. This result has been completed and generalized more recently in Ref.~\cite{Burgarth2018} through an adiabatic theorem~\cite{Messiah2017, Kato1950, ref:unity1, ref:EternalAdiabatic}, where it is shown that a quantum Zeno dynamics is obtained by repeatedly kicking the system with a cycle of general quantum operations. In the context of dynamical decoupling, this situation is described by \emph{periodic deterministic decoupling} (PDD). See Refs.~\cite{Viola2006, Santos2008} for an overview of different decoupling schemes. However, interestingly, for larger system dimensions or longer evolution times, randomized decoupling schemes seem to be more efficient in suppressing errors than their deterministic counterparts~\cite{Viola2005, Santos2006, Viola2006, Santos2008}. In its most elementary form, random dynamical decoupling is achieved by drawing the unitary kicks at random from some group at each time step. We make this notion of so-called \emph{naive random decoupling} (NRD)~\cite{Viola2005, Santos2006, Viola2006, Santos2008} rigorous later. For now, we emphasize that such a scheme is intrinsically acyclic~\cite{Viola2005}. Therefore, the previous unifications of dynamical decoupling and the quantum Zeno effect~\cite{Facchi2004, Burgarth2018} do not directly apply to the NRD case. This raises the question whether a similar consideration is still possible for random dynamical decoupling. The goal of this paper is to answer this question by using a symmetry between the system and environmental dynamics, that we call \emph{equitability of system and bath} and that provides a framework for unifying random dynamical decoupling and the quantum Zeno effect.

The idea is to study the \emph{average evolution} under the random dynamical decoupling, which can be described by the Zeno dynamics, where one frequently projects the system onto the maximally mixed state. Unfortunately, by this procedure one loses all information about the initial system state making it rather unsuitable for quantum control purposes. Rather, the goal of dynamical decoupling is to maintain the initial state. However, one can still learn something about random dynamical decoupling when investigating the Zeno dynamics induced by its average evolution. Namely, by removing the interaction Hamiltonian, the latter leads to an environmental unitary dynamics. 
As a consequence, the environmental evolution under non-averaged random dynamical decoupling has to be close to the unitary with a high probability. This statement can be lifted to the system evolution by a Schmidt-decomposition of the Choi-Jamip\l{}kowski state of the total evolution. By this virtue, explicit quantitative error bounds for the random dynamical decoupling are directly obtained from bounding the Zeno error. This shows that the decoupling efficiency of the random dynamical decoupling can be written in terms of the Zeno convergence of its average evolution, which manifests an interesting connection between the random dynamical decoupling and the quantum Zeno effect that we want to explore here in detail. 

This paper is structured as follows. We start by briefly introducing some mathematical preliminaries in Sec.~\ref{sec:Definitions}. It includes a recap on completely positive and trace-preserving (CPTP) maps in Sec.~\ref{sec:CPTP} as well as on randomized dynamical decoupling in Sec.~\ref{sec:Dynamical_Decoupling}. Section~\ref{sec:results} then gives an overview of the main results of this paper. In Sec.~\ref{sec:Zeno_av}, we discuss the convergence of the quantum Zeno limit. The discussion is split into two parts. First, we derive a refined error bound on the convergence of the quantum Zeno limit in Sec.~\ref{sec:Zeno}. We then apply this bound to the case of the average dynamical decoupling protocol to bound the convergence of its environmental evolution in Sec.~\ref{sec:av}. As we wish to shift our discourse from CPTP maps to Choi-Jamio\l{}kowski states, we add a discussion about reduced Choi-Jamio\l{}kowski states in Sec.~\ref{sec:Choi}. Finally, in Sec.~\ref{sec:conergence_tr}, we prove the main theorems presented in Sec.~\ref{sec:results}, and conclude this paper in Sec.~\ref{sec:conclusion}. Appendix~\ref{appendix:lemmas} contains some basic lemmas. The models we use in numerical simulations are specified in Appendix~\ref{appendix:numerics} for transparency reasons.

\section{Mathematical Definitions and Notation}\label{sec:Definitions}
In this section, we introduce the notation we will use throughout the paper. 
We also define what we mean by random dynamical decoupling both in the trajectory and in the average picture.
Let us start with some preliminaries on \emph{completely positive and trace-preserving (CPTP) maps}.

\subsection{Preliminaries on CPTP Maps}\label{sec:CPTP}
Dynamical decoupling protocols are described by CPTP maps, also known as \emph{quantum channels}~\cite{Nielsen2011, Wolf2012, Watrous2018}. 
For completeness and notation, we start by recapitulating some basic properties of CPTP maps. 
This subsection can be skipped by experienced readers.

Recall that a CPTP map $\mathcal{T}$ sends density operators to density operators. 
Its action on an input state $\rho$ can be written in the \emph{Kraus representation}
\begin{equation}
\mathcal{T}(\rho)=\sum_k E_k\rho E_k^\dagger, \label{Kraus-rep}
\end{equation}
where the Kraus operators $E_k$ satisfy $\sum_{k}E_k^\dagger E_k=\mathbb{1}$. 
Now, in a fixed basis $\{\ket{e_i}\}_{i=1,\ldots,d}$ of a $d$-dimensional system, a density operator $\rho$ has a matrix representation
\begin{equation}
\rho =\sum_{i,j} r_{ij}\ket{e_i}\bra{e_j}.
\end{equation}
By concatenating each row of this matrix, we construct an equivalent vector $\vecket{\rho}$ in a $d^2$-dimensional vector space as
\begin{equation}
\vecket{\rho}=\sum_{i,j} r_{ij} \ket{e_i}\otimes\ket{e_j}.
\end{equation}
This procedure falls under the name of \emph{row-vectorization}~\cite{ref:VectorizationHavel, Wood2015, Watrous2018} and is a basis-dependent linear isomorphism $\rho\mapsto\vecket{\rho}$. 
It is closely related to the Hilbert-Schmidt inner product of operators as
\begin{equation}
(A|B)
=\tr(A^\dag B).
\end{equation}
The row-vectorization $A\mapsto\vecket{A}$ can also be done by
\begin{equation}
\vecket{A}=(A\otimes\mathbb{1})|\mathbb{1}),\label{eq:vectorization}
\end{equation}
where $|\mathbb{1})=\sum_i\ket{e_i}\otimes\ket{e_i}$ is a maximally entangled state up to normalization, defined on the given basis $\{\ket{e_i}\}$.
In this framework, a CPTP map $\mathcal{T}:\rho\mapsto\rho'$ is a $d^2\times d^2$ matrix $\hat{\mathcal{T}}:\vecket{\rho}\mapsto\vecket{\rho'}$, acting on vectorized density matrices. 
The action of CPTP maps then just becomes matrix multiplication. This matrix $\hat{\mathcal{T}}$ is a \emph{matrix representation} of the CPTP map $\mathcal{T}$. 
In the following, we denote the matrix representation corresponding to a CPTP map with a hat. 
In order to handle vectorized operators, there is a useful theorem called Roth's lemma~\cite{Ward1999}: for $d\times d$ matrices $A$, $B$, and $C$, we have
\begin{equation}
\vecket{ABC}=(A\otimes C\transpose)\vecket{B},\label{Roth}
\end{equation}
where $\transpose$ denotes the transposition of a matrix with respect to the given basis. 
Using Roth's lemma and the linearity of the vectorization, we get a matrix representation $\hat{\mathcal{T}}$ of the CPTP map $\mathcal{T}$ in terms of the Kraus operators,
\begin{equation}
\hat{\mathcal{T}}=\sum_k E_k\otimes \overline{E_k},\label{Kraus}
\end{equation}
where the bar denotes the entry-wise complex conjugation in the given basis. 
If $\mathcal{T}$ is unitary, the matrix representation~\eqref{Kraus} reduces to
\begin{equation}
\hat{\mathcal{T}}=U\otimes \overline{U},
\end{equation}
with a unitary $U$.

There is another useful representation of CPTP maps, namely~the \emph{Choi-Jamio\l{}kowski state} $\Lambda$ defined by
\begin{equation}
\Lambda=(\mathcal{T}\otimes\mathbb{I})\!\left(\frac{1}{d}|\mathbb{1})(\mathbb{1}|\right),\label{eq:Choi-state}
\end{equation}
where $\mathbb{I}$ represents the identity map.
It is normalized as $\tr\Lambda=1$ for a TP map $\mathcal{T}$, and the map $\mathcal{T}$ is CP iff $\Lambda\ge0$.
Let us consider the spectral decomposition of the Choi-Jamio\l{}kowski state $\Lambda$ of a CPTP map $\mathcal{T}$,
\begin{equation}
\Lambda=\sum_k \lambda_k|v_k)(v_k|.
\label{CJ-spectral}
\end{equation}
We have $\lambda_k\ge0$, $\forall k$, and $\sum_k \lambda_k = 1$.
The eigenvectors $|v_k)$ are orthonormalized as $(v_k|v_\ell)=\delta_{k\ell}$, and provide the Kraus operators
\begin{equation}
E_k=\sqrt{d\lambda_k}\,v_k\label{Kraus-choice}
\end{equation}
of a Kraus representation (\ref{Kraus-rep}) of the given CPTP map $\mathcal{T}$~\cite[Proposition~2.20]{Watrous2018}.
While the Kraus representation is not unique and the set of Kraus operators given by (\ref{Kraus-choice}) is not the only choice, this will be a convenient pick for our analysis in the following.
If only one of the $d^2$ eigenvalues $\lambda_k$ of $\Lambda$ is nonvanishing, namely,~if the Choi-Jamio\l{}kowski state $\Lambda$ is a pure state, the map $\mathcal{T}$ is unitary.

We will be interested in comparing different CPTP maps.
For instance, we will wish to estimate the distance of a CPTP map $\mathcal{T}$ from a unitary evolution in the Zeno limit. 
In order to be able to talk about the distance between CPTP maps, we will need some distance measures, which can be defined via norms of maps. 
We will use Schatten $p$-norms defined for matrices $A$ by
\begin{equation}
\|A\|_p
=\left(
\tr[(A^\dagger A)^{\frac{p}{2}}]
\right)^{\frac{1}{p}},
\end{equation}
where $\dagger$ denotes the adjoint with respect to the Hilbert-Schmidt scalar product.
A useful norm for our purposes will be the Frobenius norm 
\begin{equation}
\Fnorm{A}=\sqrt{\tr(A^\dagger A)}.
\label{eq:Frobenius_norm}
\end{equation}
We will also use the operator norm $\|A\|_\infty$,
which gives the largest singular value of $A$. 
For density operators, the trace norm
\begin{equation}
\|A\|_1=\tr\sqrt{A^\dag A}=\tr|A|
\end{equation}
is useful.
Note that in finite-dimensional vector spaces all norms are equivalent~\cite{Horn2012}. 
In particular, we will use the following norm equivalences~\cite[Eqs.~(1.168) and~(1.169)]{Watrous2018} for the Choi-Jamio\l{}kowski state $\Lambda$, which is a $d^2\times d^2$ matrix for a CPTP map $\mathcal{T}$ acting on a $d$-dimensional system,
\begin{align}
\infnorm{\Lambda}&\leq \Fnorm{\Lambda}\leq d\infnorm{\Lambda},
\label{eqn:NormEquivalence}
\\
	\|\Lambda\|_\infty
	&\le
	\|\Lambda\|_1
	\le
	d^2\|\Lambda\|_\infty,
\label{eqn:NormEquiv1inf}
\\
	\|\Lambda\|_2
	&\le
	\|\Lambda\|_1
	\le
	d\|\Lambda\|_2.
	\label{eq:equivalence_tr_fro}
\end{align}

The operator norm $\|\hat{\mathcal{T}}\|_\infty$ of a matrix representation $\hat{\mathcal{T}}$ of a map $\mathcal{T}$ can also be induced by the Frobenius norm on the vector space $\mathscr{A}$ of matrices on which the map $\mathcal{T}$ acts,
\begin{equation}
\infnorm{\hat{\mathcal{T}}} =\|\mathcal{T}\|_{2\to2},\qquad
\|\mathcal{T}\|_{p\to q}=\sup_{A\in\mathscr{A}}\frac{\|\mathcal{T}(A)\|_q}{\|A\|_p}.\label{eq:22norm_infnorm}
\end{equation}
To bound the convergence to a unitary via dynamical decoupling, the Frobenius norm $\|\hat{\mathcal{T}}\|_2$ will be a natural one, since it is easily related to the purity $\pur(\Lambda)=\tr(\Lambda^2)$ of the Choi-Jamio\l{}kowski state $\Lambda$, which is a measure of the unitarity of the corresponding map $\mathcal{T}$.
Indeed, the Frobenius norm of the Kraus representation (\ref{Kraus}) yields
\begin{align}
\|\hat{\mathcal{T}}\|_2^2
&=
\left\|
\sum_kE_k\otimes\overline{E_k}
\right\|_2^2
\nonumber\\
&=
\tr\!\left[
\left(
\sum_kE_k^\dag\otimes\overline{E_k}^\dag
\right)
\left(
\sum_\ell E_\ell\otimes\overline{E_\ell}
\right)
\right]
\nonumber\\
&=\sum_{k,\ell}|(E_k|E_\ell)|^2
\nonumber\\
&=d^2\sum_{k,\ell}\lambda_k\lambda_\ell|(v_k|v_\ell)|^2
\nonumber\\
&=d^2\sum_k\lambda_k^2
\nonumber\\
&=d^2\pur(\Lambda),\label{eq:Fnorm_purity}
\end{align}
where we have used the Kraus operators (\ref{Kraus-choice}) and the orthonormality $(v_k|v_\ell)=\delta_{k\ell}$ of the eigenvectors $|v_k)$ of $\Lambda$ in (\ref{CJ-spectral}).
To bound the Zeno limit, the operator norm $\infnorm{\hat{\mathcal{T}}}$ will prove useful, as we will have to deal with Hermitian projection operators, whose operator norms are simply $1$. Even though this property also holds for the $1\rightarrow 1$ norm $\pnorm{\mathcal{T}}{1\rightarrow 1}$, which is induced by the trace norm, the operator norm has the advantage that it directly relates to the Frobenius norm via~\eqref{eq:22norm_infnorm}. This makes it more attractive for a Zeno bound in achieving the goal of unifying the Zeno dynamics with random dynamical decoupling.
In the literature, a commonly used norm of maps is the the diamond norm~\cite[Sec.~3.3]{Watrous2018}
\begin{equation}
	\pnorm{\mathcal{T}}{\diamond}
	=
	\sup_{\|A\|_1=1}
	\|
	(\mathcal{T}\otimes\mathbb{I})(A)
	\|_1,
\end{equation}
which is stable under tensoring additional subsystems. We will also convert our error bounds to this norm using a norm equivalence established in Lemma~\ref{lemma:diamond_choi} in Appendix~\ref{appendix:lemmas}\@.

\subsection{Random Dynamical Decoupling}\label{sec:Dynamical_Decoupling}
In this subsection, we define the protocols of dynamical decoupling introduced in Sec.~\ref{sec:Introduction}.

Consider a bipartite Hilbert space $\mathscr{H}=\mathscr{H}_1\otimes\mathscr{H}_2$ with $d_1=\dim\mathscr{H}_1<\infty$ and $d_2=\dim\mathscr{H}_2<\infty$. We model an open quantum system with the Hamiltonian
\begin{equation}
H=H_1\otimes\mathbb{1}_2+\mathbb{1}_1\otimes H_2+H_{12}, 
\label{Hamilton_decomp}
\end{equation}
where $H_1$ is the Hamiltonian of system 1 (``system''), $H_2$ describes system 2 (``environment'' or ``bath''),
and $H_{12}=\sum_i w_i h_1^{(i)}\otimes h_2^{(i)}$ is the interaction, i.e., $H_1$, $H_2$, $h_1^{(i)}$, and $h_2^{(i)}$ are Hermitian operators, and $w_i\in\mathbb{R}$.
Without loss of generality, we neglect global phases and furthermore set $\tr H_1=\tr H_2=\tr h_1^{(i)}=\tr h_2^{(i)}=0$, $\forall i$.
Using an operator version of the Schmidt decomposition~\cite{Tyson2003}, such decomposition is always possible for finite-dimensional systems, and the operators involved in the decomposition are orthogonal with respect to the Hilbert-Schmidt inner product.
The total evolution under this Hamiltonian $H$ is then given by a unitary $\rme^{-\rmi t\mathcal{H}}$ with $\mathcal{H}=[H,{}\bullet{}]$. 
In the following, we will use calligraphic symbols to denote maps acting on operators on the Hilbert space $\mathscr{H}$.

The goal of a dynamical decoupling protocol is to quickly rotate $\mathscr{H}_1$ in order to average out the system-bath interaction $H_{12}$. In the standard periodic deterministic dynamical decoupling protocol~\cite{Viola1999, Viola1999a, Viola2003, Khodjasteh2005}, we apply unitaries on system 1 in a fixed order. The unitaries are taken from an irreducible representation $\mathscr{V}=\{V_1,\ldots,V_{|\mathscr{V}|}\}$ of a finite group, and the total evolution is given by
\begin{equation}
\mathcal{E}_m(t)
=\left(
(\mathcal{V}_{|\mathscr{V}|}^\dagger\rme^{-\rmi\frac{t}{m|\mathscr{V}|}\mathcal{H}}\mathcal{V}_{|\mathscr{V}|})
\cdots
(\mathcal{V}_2^\dagger\rme^{-\rmi\frac{t}{m|\mathscr{V}|}\mathcal{H}}\mathcal{V}_2)
(\mathcal{V}_1^\dagger\rme^{-\rmi\frac{t}{m|\mathscr{V}|}\mathcal{H}}\mathcal{V}_1)
\right)^m,\quad m\in\mathbb{N},
\label{eqn:StandardDD}
\end{equation}
where $\mathcal{V}_i=(V_i\otimes\mathbb{1}_2){}\bullet{}(V_i^\dag\otimes\mathbb{1}_2)$ with $\mathbb{1}_2$ the identity operator on $\mathscr{H}_2$. Observe that this equation is written in terms of quantum channels, which will provide a unified picture with random dynamical decoupling.

In the random dynamical decoupling protocol, on the other hand, we choose a unitary $V_i^{(j)}$ randomly (i.i.d.) from $\mathscr{V}$ for each step and proceed as
\begin{equation}
\mathcal{E}^{(j)}_n(t)
=
\mathcal{V}_{n+1}^{(j)}
\rme^{-\rmi\frac{t}{n}\mathcal{H}}
\mathcal{V}_n^{(j)}
\cdots\,
\rme^{-\rmi\frac{t}{n}\mathcal{H}}
\mathcal{V}_2^{(j)}
\rme^{-\rmi\frac{t}{n}\mathcal{H}}
\mathcal{V}_1^{(j)},\quad n\in\mathbb{N},
\label{eqn:TrajectoryProtocol}
\end{equation}
where $\mathcal{V}_i^{(j)}=(V_i^{(j)}\otimes\mathbb{1}_2){}\bullet{}(V_i^{(j)\dag}\otimes\mathbb{1}_2)$.
A sequence of $(n+1)$ randomly sampled unitaries 
$\{V_1^{(j)},\ldots,V_{n+1}^{(j)}\}$ 
characterizes a ``trajectory,'' and there are $|\mathscr{V}|^{n+1}$ different trajectories, which are labeled by $j=1,\ldots,|\mathscr{V}|^{n+1}$.
Each of the $|\mathscr{V}|^{n+1}$ trajectories is realized by the sequence of random i.i.d.\ samplings.

If we do not know which random unitaries $V_i^{(j)}$ are drawn from the unitary group $\mathscr{V}$, a good description of the protocol is provided by the average picture.
It is obtained by averaging the trajectory protocols $\mathcal{E}_n^{(j)}(t)$ over all possible trajectory realizations.
The channel representing the average picture is given by
\begin{align}
\mathcal{E}_n^\mathrm{av}(t)
&=\E[\mathcal{E}_n^{(j)}(t)]
\nonumber\\
&=
\frac{1}{|\mathscr{V}|^{n+1}}
\sum_{V_{n+1}^{(j)}\in\mathscr{V}}
\sum_{V_{n}^{(j)}\in\mathscr{V}}
\cdots
\sum_{V_1^{(j)}\in\mathscr{V}}
\mathcal{V}_{n+1}^{(j)}
\rme^{-\rmi\frac{t}{n}\mathcal{H}}
\mathcal{V}_n^{(j)}
\cdots\,
\rme^{-\rmi\frac{t}{n}\mathcal{H}}
\mathcal{V}_1^{(j)}
\nonumber\\
&=(\mathcal{D}\rme^{-\rmi\frac{t}{n}\mathcal{H}}\mathcal{D})^n,
\label{eq:average_protocol}
\end{align}
where $\mathcal{D}$ is the group average
\begin{equation}
\mathcal{D}(A)
=\frac{1}{|\mathscr{V}|}\sum_{V\in\mathscr{V}}(V\otimes\mathbb{1}_2)A(V^\dagger\otimes\mathbb{1}_2)
=\frac{1}{d_1}\mathbb{1}_1\otimes(\tr_1A)
\label{DDprotocol}
\end{equation}
over system 1. By irreducibility of $\mathscr{V}$, it projects
system 1 onto the maximally mixed state $\mathbb{1}_1/d_1.$
Note that $\mathcal{D}$ is CPTP~\cite{Wolf2012} and, furthermore, a Hermitian projection, as is clear from its matrix representation 
\begin{equation}
\hat{\mathcal{D}}=\frac{1}{d_1}\vecket{\mathbb{1}_1}\vecbra{\mathbb{1}_1}\otimes\hat{\mathbb{I}}_{22'},
\end{equation}
where $|\mathbb{1}_1)$ is the row-vectorization of the identity operator $\mathbb{1}_1$ on $\mathscr{H}_1$ to an extended Hilbert space $\mathscr{H}_1\otimes\mathscr{H}_{1'}$ of system 1, while $\hat{\mathbb{I}}_{22'}$ is the identity on another extended Hilbert space $\mathscr{H}_2\otimes\mathscr{H}_{2'}$ of system 2.
The average channel $\mathcal{E}_n^\mathrm{av}(t)$ in~\eqref{eq:average_protocol} is formally equivalent to the standard quantum Zeno dynamics by frequent projective measurements~\cite{Misra1977, Home1997, Facchi2001a, Facchi2001, ref:PaoloSaverio-QZEreview-JPA}. 
That is why we will estimate an error bound on the quantum Zeno limit, which immediately gives a bound on the convergence of $\mathcal{E}_n^\mathrm{av}(t)$ as $n\rightarrow\infty$. This bound can then be used to bound the convergence of the trajectory protocol $\mathcal{E}_n^{(j)}(t)$ in the decoupling limit $n\to\infty$.

\section{Main Results}\label{sec:results}
Our main objective is to establish an explicit connection between the quantum Zeno dynamics and the decoupled dynamics by a random dynamical decoupling sequence. As we have seen above, the average dynamical decoupling protocol is a manifestation of the quantum Zeno procedure. We will show that the decoupling efficiency of the random dynamical decoupling protocol $\mathcal{E}_n^{(j)}(t)$ is determined by the convergence speed of the average protocol $\mathcal{E}_n^\mathrm{av}(t)$, which is ruled by the Zeno limit.
We hence first provide an explicit bound on the quantum Zeno limit via frequent projective measurements. 
The Zeno limit of the average protocol $\mathcal{E}_n^\mathrm{av}(t)$ yields quantum Zeno dynamics, in which system 2 evolves unitarily.
The fact that the evolution of system 2 converges to a unitary on average implies that almost all trajectory evolutions of system 2 in the trajectory picture $\mathcal{E}_n^{(j)}(t)$ also converge to the unitary in the decoupling limit.
Since the total evolution $\mathcal{E}_n^{(j)}(t)$ in the trajectory picture is unitary, this further implies that the evolution of system 1 in the trajectory picture also converges to a unitary, and the evolutions of systems 1 and 2 are decoupled.

To prove these statements in mathematically rigorous ways, we look at the Choi-Jamio\l{}kowski states of the evolutions. 
In this section, we present and discuss main theorems, but defer their proofs to later sections.
Let $\Lambda^{(j)}(n)$ be the Choi-Jamio\l{}kowski state of the random trajectory dynamical decoupling protocol $\mathcal{E}_n^{(j)}(t)$ defined in (\ref{eqn:TrajectoryProtocol}), i.e.,
\begin{equation}
\Lambda^{(j)}(n)
=[\mathcal{E}_n^{(j)}(t)\otimes\mathbb{I}_{1'2'}]\!\left(\frac{1}{d}|\mathbb{1}_{12})(\mathbb{1}_{12}|\right),
\end{equation}
where $d=d_1d_2$ is the dimension of the total system.
By linearity, the Choi-Jamio\l{}kowski state $\Lambda^\mathrm{av}(n)$ of the random dynamical decoupling protocol in the average picture $\mathcal{E}_n^\mathrm{av}(t)$ is obtained by taking the average of the Choi-Jamio\l{}kowski states $\Lambda^{(j)}(n)$ of the trajectory picture over all trajectories,
\begin{equation}
\Lambda^\mathrm{av}(n)=\E[\Lambda^{(j)}(n)].
\end{equation}
Taking the trace over the doubled Hilbert space $\mathscr{H}_1\otimes\mathscr{H}_{1'}$ of system 1, we get a reduced Choi-Jamio\l{}kowski state for system 2,
\begin{equation}
\Lambda_2^\mathrm{av}(n)
=\tr_{11'}\Lambda^\mathrm{av}(n)
=\Lambda_{2,\mathbb{1}_1/d_1}^\mathrm{av}(n).
\end{equation}
This is the Choi-Jamio\l{}kowski state of the reduced dynamics $\mathcal{E}_{2,\mathbb{1}_1/d_1}^\mathrm{av}(n)$ of $\mathcal{E}_n^\mathrm{av}(t)$ for system 2 when the initial state of system 1 is the maximally mixed state $\mathbb{1}_1/d_1$ (see Sec.~\ref{sec:Choi}).

\subsection{Convergence of the Average Evolution of System 2}\label{sec:Av2}
We first show that the reduced Choi-Jamio\l{}kowski state $\Lambda_2^\mathrm{av}(n)$ converges to a pure state, i.e., the corresponding reduced dynamics of system 2 converges to a unitary in the decoupling limit $n\to\infty$.
To bound the convergence, we prove the following theorem on the standard Zeno limit via frequent projective measurements.
\begin{theorem}[Explicit error bound on the quantum Zeno limit]
\label{thm:convergence_average}
Let $H=H^\dagger$ be a Hermitian operator and let $P=P^2=P^\dagger$ be a Hermitian projection. 
Then, for any $t\ge0$ and $n\in\mathbb{N}$, we have
\begin{equation}
\|(\rme^{-\rmi \frac{t}{n}H}P)^n-\rme^{-\rmi tPHP}P\|_\infty
\le\frac{t}{n}\|H\|_\infty + \frac{t^2}{n}\|H\|_\infty^2.
\label{Convergence1}
\end{equation}
\end{theorem}
\begin{proof}
This is proved in Sec.~\ref{sec:Zeno}.
\end{proof}
\noindent
Notice that the bound in Theorem~\ref{thm:convergence_average} shows an $\O(1/n)$ behavior for $n\rightarrow\infty$, which is consistent with numerical simulations. See Fig.~\ref{fig:Zeno_full}. It takes a particularly simple form and therefore has a few advantages over previously known error bounds on the quantum Zeno effect. For instance, in comparison to Ref.~\cite{Dominy2013}, this bound looks much simpler. The bound in Ref.~\cite[Theorem~1]{Becker2020} only shows an $\O(n^{-2/3})$ scaling. Reference~\cite[Lemma~4]{Moebus2019} studies an error bound in a more general setting, which still depends on an undetermined constant. This result has recently been improved in Ref.~\cite[Theorem 3.1]{Moebus2021} to an $O(1/n)$ scaling (also see lemmas 5.2, 5.5 and 5.6 therein). Here, the authors receive a similar bound to ours in a more general setting. Nevertheless, the bound presented in Theorem~\ref{thm:convergence_average} is slightly tighter and emerges from a simpler proof.

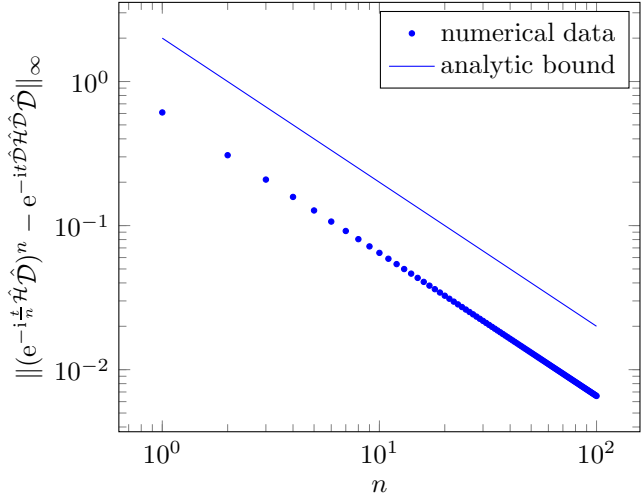
\begin{figure}
\centering
	\begin{tikzpicture}[mark size={1}, scale=1]
	\begin{axis}[
	xmode=log,
	ymode=log,
	xlabel={$n$},
	ylabel={$\|(\rme^{-\rmi\frac{t}{n}\hat{\mathcal{H}}} \hat{\mathcal{D}})^n-\rme^{-\rmi t \hat{\mathcal{D}}\hat{\mathcal{H}}\hat{\mathcal{D}}} \hat{\mathcal{D}}\|_\infty$},
	x post scale=1,
	y post scale=1,
	transpose legend,
	legend columns = 2,
	]
	\addplot[color=blue, only marks] table[x=n, y=DU, col sep=comma]{Dav.csv};
	\addplot[domain=1:100, color=blue, samples=100]{2/x};
	\legend{numerical data, analytic bound}
	\end{axis}
	\end{tikzpicture}
  	\label{fig:Dav}
\caption{
Comparison of the Zeno bound in Theorem~\ref{thm:convergence_average} with a numerical simulation. The projection is chosen to be $\hat{\mathcal{D}}=\frac{1}{d_1}\vecket{\mathbb{1}_1}\vecbra{\mathbb{1}_1}\otimes\hat{\mathbb{I}}_{22'}$ from the random average dynamical decoupling scheme. See Sec.~\ref{sec:av}. The Hamiltonian is a generically chosen traceless Hermitian matrix on two qubits, which is concretely specified in Appendix~\ref{appendix:numerics}\@. The total evolution time $t$ is set so that $T=t\|\hat{\mathcal{H}}\|_\infty=1$.
}
\label{fig:Zeno_full}
\end{figure}

Theorem~\ref{thm:convergence_average} allows us to bound the convergence of the Choi-Jamio\l{}kowski state $\Lambda_{2,\sigma_1}^\mathrm{av}(n)$ of the reduced dynamics $\mathcal{E}_{2,\sigma_1}^\mathrm{av}(n)$ of $\mathcal{E}_n^\mathrm{av}(t)$ for system 2 when the initial state of system 1 is $\sigma_1$.
\begin{proposition}[Distance of the average evolution of system 2 from the Zeno dynamics]
\label{prop:Av2D}
\begin{gather}
\left\|
\Lambda^\mathrm{av}_{2,\sigma_1}(n)
-
\frac{1}{d_2}|\rme^{-\rmi tH_2})(\rme^{-\rmi tH_2}|
\right\|_2
\le
\frac{1}{n}\sqrt{d_1}\,\|\sigma_1\|_2T^2,\label{eq:DU2av}\\
\|
\mathcal{E}_{2,\sigma_1}^\mathrm{av}(n)
-
\rme^{-\rmi t\mathcal{H}_2}
\|_\diamond
\le
\frac{1}{n}\sqrt{d_1}d_2^2\|\sigma_1\|_2T^2,\label{eq:DU2av_diamond}
\end{gather}
for any $t\ge0$ and $n\in\mathbb{N}$, where
\begin{equation}
T=t\|\hat{\mathcal{H}}\|_\infty,
\end{equation}
with $\hat{\mathcal{H}}$ the matrix representation of the generator $\mathcal{H}=[H,{}\bullet{}]$.
\end{proposition}
\begin{proof}
This is proved in Sec.~\ref{sec:av}.
\end{proof}
\begin{remark}
	Notice that the decoupling bound for system 2 given in Proposition~\ref{prop:Av2D} depends on the quantity $T^2/n$, whereas the bound on the Zeno limit in Theorem~\ref{thm:convergence_average} reads $T(1+T)/n$. This is due to the fact that we have introduced an additional unitary $\mathcal{V}_{n+1}$ at the end of the decoupling sequence in~(\ref{eqn:TrajectoryProtocol}). Then, the averaged random dynamical decoupling protocol relates to the Zeno dynamics introduced in Proposition~\ref{lemma:Dav_convergence} instead of the one shown in Theorem~\ref{thm:convergence_average}. See~(\ref{eq:average_protocol}). 
		If the last unitary $\mathcal{V}_{n+1}$ is removed from the sequence in~(\ref{eqn:TrajectoryProtocol}), then Eq.~(\ref{eq:average_protocol}) becomes $\mathcal{E}_n^\mathrm{av}(t)=(\rme^{-\rmi\frac{t}{n}\mathcal{H}}\mathcal{D})^n$. This leads to a change with $T^2$ replaced by $T(1+T)$ in all bounds.
\end{remark}
The convergence of the reduced dynamics of system 2 to the unitary $\rme^{-\rmi tH_2}$ can also be verified by the convergences of the spectral norm $\|\Lambda^\mathrm{av}_{2,\sigma_1}(n)\|_\infty$ and of the purity $\pur\bigl(\Lambda^\mathrm{av}_{2,\sigma_1}(n)\bigr)$ to $1$.
\begin{proposition}[Purity of the average evolution of system 2]
\label{prop:Av2P}
\begin{equation}
\sqrt{\mathtt{P}\bigl(\Lambda^\mathrm{av}_{2,\sigma_1}(n)\bigr)}
\ge
\|\Lambda^\mathrm{av}_{2,\sigma_1}(n)\|_\infty
\ge
\frac{1}{d_2}
(\rme^{-\rmi tH_2}|\Lambda^\mathrm{av}_{2,\sigma_1}(n)|\rme^{-\rmi tH_2})
\ge1-\frac{1}{n}\sqrt{d_1}\,\|\sigma_1\|_2T^2,
\end{equation}
for any $t\ge0$ and $n\in\mathbb{N}$.
\end{proposition}
\begin{proof}
This is proved in Sec.~\ref{sec:av}.
\end{proof}

\subsection{Convergence of the Trajectory Evolution of System 2}\label{sec:Tr2}
The fact that the average evolution of system 2 converges to a unitary implies that almost all the trajectory evolutions of system 2 converge to a unitary.
We prove the following two theorems for the Choi-Jamio\l{}kowski state $\Lambda_{2,\sigma_1}^{(j)}(n)$ of the reduced dynamics $\mathcal{E}_{2,\sigma_1}^{(j)}(n)$ of the trajectory protocol $\mathcal{E}_n^{(j)}(t)$ for system 2 when the initial state of system 1 is $\sigma_1$.
\begin{theorem}[Purity of the trajectory evolution of system 2]
\label{thm:Tr2P}
\begin{equation}
\sqrt{
\mathbb{E}\!\left[
\mathtt{P}\bigl(\Lambda_{2,\sigma_1}^{(j)}(n)\bigr)
\right]
}
\ge
\mathbb{E}\Bigl[\|\Lambda_{2,\sigma_1}^{(j)}(n)\|_\infty\Bigr]
\ge
\mathbb{E}\!\left[
\frac{1}{d_2}
(\rme^{-\rmi tH_2}|\Lambda_{2,\sigma_1}^{(j)}(n)|\rme^{-\rmi tH_2})
\right]
\ge
1-\frac{1}{n}\sqrt{d_1}\,\|\sigma_1\|_2T^2,
\end{equation}
for any $t\ge0$ and $n\in\mathbb{N}$.
\end{theorem}
\begin{proof}
This is proved in Sec.~\ref{sec:proofTr2}.	
\end{proof}
\noindent
The $\order{1/n}$ scaling of the decoupling limit $n\rightarrow\infty$ is observed in numerical simulations. See the green points in Figs.~\ref{fig:Purity} and~\ref{fig:OpNorm}.
\begin{theorem}[Distance of the trajectory evolution of system 2 from the Zeno dynamics]
\label{thm:Tr2D}
\begin{gather}
\mathbb{E}\!\left[
\left\|
\Lambda_{2,\sigma_1}^{(j)}(n)
-
\frac{1}{d_2}|\rme^{-\rmi tH_2})(\rme^{-\rmi tH_2}|
\right\|_2
\right]
\le
\sqrt{\frac{2}{n}\sqrt{d_1}\,\|\sigma_1\|_2T^2},
\label{eq:DU2tr}
\\
\mathbb{E}\Bigl[
\|
\mathcal{E}_{2,\sigma_1}^{(j)}
-
\rme^{-\rmi t\mathcal{H}_2}
\|_\diamond
\Bigr]
\le
d_2^2\sqrt{\frac{2}{n}\sqrt{d_1}\,\|\sigma_1\|_2T^2},
\label{eq:DU2tr_diamond}
\end{gather}
for any $t\ge0$ and $n\in\mathbb{N}$.
\end{theorem}
\begin{proof}
This is proved in Sec.~\ref{sec:proofTr2}.	
\end{proof}
\noindent
The distance to the unitary of the Zeno dynamics, which is specified by the bath Hamiltonian $H_2$, shrinks to zero in the decoupling limit $n\rightarrow\infty$ as $\order{1/\sqrt{n}}$, which is consistent with numerical simulations. See the green points in Fig.~\ref{fig:DU_Choi}.

\begin{figure}
\centering
	\begin{subfigure}[r]{.45\textwidth}
  	\begin{tikzpicture}[mark size={0.6}, scale=1]
	\begin{axis}[
	xmode=log,
	ymode=log,
	xlabel={$n$},
	ylabel={$1-\mathbb{E}[\pur(\Lambda)]$},
	x post scale=0.8,
	y post scale=0.8,
	]
	\addplot[color=blue, only marks] table[x=n, y=P, col sep=comma] {Purity_1_sigma.csv};
	\addplot[domain=3.54655:100, color=blue, samples=100]{1-(1-2/x)^2};
	\addplot[color=red, only marks] table[x=n, y=P, col sep=comma] {Purity_1.csv};
	\addplot[domain=1.77328:100, color=red, samples=100]{1-(1-1/x)^2};
	\addplot[color=green, only marks] table[x=n, y=P, col sep=comma] {Purity_2_sigma.csv};
	\addplot[domain=2.50779:100, color=green, samples=100]{1-(1-sqrt(2)*1/x)^2};
	\end{axis}
	\end{tikzpicture}
  	\caption{Purity of the reduced Choi-Jamio\l{}kowski state of the trajectory evolution under random dynamical decoupling. Blue: Theorem~\ref{thm:Tr1P} with $\sigma_2=\ket{0}\bra{0}$. Red: Theorem~\ref{thm:Tr1P} with $\sigma_2=\mathbb{1}_2/d_2$. Green: Theorem~\ref{thm:Tr2P} with $\sigma_1=\ket{0}\bra{0}$.}
  	\label{fig:Purity}
	\end{subfigure}
	\hfill
	\begin{subfigure}[l]{.45\textwidth}
  	\begin{tikzpicture}[mark size={0.6}, scale=1]
	\begin{axis}[
	xmode=log,
	ymode=log,
	ylabel near ticks,
	yticklabel pos=right,
	xlabel={$n$},
	ylabel={$1-\mathbb{E}[\infnorm{\Lambda}]$},
	x post scale=0.8,
	y post scale=0.8,
	]
	\addplot[color=blue, only marks] table[x=n, y=OpNorm, col sep=comma] {OpNorm_1_sigma.csv};
	\addplot[domain=1:100, color=blue, samples=100]{2/x};
	\addplot[color=red, only marks] table[x=n, y=OpNorm, col sep=comma] {OpNorm_1.csv};
	\addplot[domain=1:100, color=red, samples=100]{1/x};
	\addplot[color=green, only marks] table[x=n, y=OpNorm, col sep=comma] {OpNorm_2_sigma.csv};
	\addplot[domain=1:100, color=green, samples=100]{sqrt(2)/x};
	\end{axis}
	\end{tikzpicture} 
  	\caption{Operator norm of the reduced Choi-Jamio\l{}kowski state of the trajectory evolution under random dynamical decoupling. Blue: Theorem~\ref{thm:Tr1P} with $\sigma_2=\ket{0}\bra{0}$. Red: Theorem~\ref{thm:Tr1P} with $\sigma_2=\mathbb{1}_2/d_2$. Green: Theorem~\ref{thm:Tr2P} with $\sigma_1=\ket{0}\bra{0}$.}
  	\label{fig:OpNorm}
	\end{subfigure}
	\vspace{0.3cm}
	\newline
	\vspace{0.3cm}
	\begin{subfigure}[r]{.45\textwidth}
	\begin{tikzpicture}[mark size={0.6}, scale=1]
	\begin{axis}[
	xmode=log,
	ymode=log,
	xlabel={$n$},
	ylabel={$\mathbb{E}[\Fnorm{\Lambda-\frac{1}{d}|v)(v|}]$},
	x post scale=0.8,
	y post scale=0.8,
	]
	\addplot[color=blue, only marks] table[x=n, y=DU, col sep=comma] {DU_Choi_1_sigma.csv};
	\addplot[domain=1:100, color=blue, samples=100]{2/x};
	\addplot[color=red, only marks] table[x=n, y=DU, col sep=comma] {DU_Choi_1.csv};
	\addplot[domain=1:100, color=red, samples=100]{1/x};
	\addplot[color=green, only marks] table[x=n, y=DU, col sep=comma] {DU_Choi_2_sigma.csv};
	\addplot[domain=1:100, color=green, samples=100]{sqrt(2*sqrt(2)/x)};
	\end{axis}
	\end{tikzpicture}
	\caption{Distances to a pure Choi-Jamio\l{}kowski state of the reduced Choi-Jamio\l{}kowski state of the trajectory evolution under random dynamical decoupling. Blue: Theorem~\ref{thm:Tr1D} with $\sigma_2=\ket{0}\bra{0}$. Red: Theorem~\ref{thm:Tr1D} with $\sigma_2=\mathbb{1}_2/d_2$. Green: Theorem~\ref{thm:Tr2D} with $\sigma_1=\ket{0}\bra{0}$.}
  	\label{fig:DU_Choi}
	\end{subfigure}
	\hfill
	\begin{subfigure}[l]{.45\textwidth}
	\begin{tikzpicture}[mark size={0.6}, scale=1]
	\begin{axis}[
	xmode=log,
	ymode=log,
	ylabel near ticks,
	yticklabel pos=right,
	xlabel={$n$},
	ylabel={$\mathbb{E}[\Fnorm{\hat{\mathcal{E}} - \hat{\mathcal{U}}}]$},
	x post scale=0.8,
	y post scale=0.8,
	]
	\addplot[color=blue, only marks] table[x=n, y=DU, col sep=comma] {DU_Super_1_sigma.csv};
	\addplot[domain=3.54655:100, color=blue, samples=100]{4*sqrt(1-(1-2/x)^4)};
	\addplot[color=red, only marks] table[x=n, y=DU, col sep=comma] {DU_Super_1.csv};
	\addplot[domain=1.77328:100, color=red, samples=100]{4*sqrt(1-(1-1/x)^4)};
	\addplot[color=green, only marks] table[x=n, y=DU, col sep=comma] {DU_Super_2_sigma.csv};
	\addplot[domain=2.50779:100, color=green, samples=100]{4*sqrt(1-(1-sqrt(2)/x)^4)};
	\end{axis}
	\end{tikzpicture}
	\caption{Distance to unitary of the reduced trajectory evolution under random dynamical decoupling. Blue: Corollary~\ref{cor:Tr1D} with $\sigma_2=\ket{0}\bra{0}$. Red: Corollary~\ref{cor:Tr1D} with $\sigma_2=\mathbb{1}_2/d_2$. Green: Theorem~\ref{thm:Tr2P} combined with Proposition~\ref{thm:dist_unitary_purity} with $\sigma_1=\ket{0}\bra{0}$.\\}
  	\label{fig:DU_Super}
	\end{subfigure}
	\begin{tikzpicture}
	\begin{customlegend}[legend columns=3, legend cell align={left}, legend style={align=center, draw, column sep=3.5ex, row sep=0.4em},
        legend entries={\small$\Lambda^{(j)}_{1,\ket{0}\bra{0}}(n)$,
                        \small$\Lambda^{(j)}_{1}(n)=\Lambda^{(j)}_{2}(n)$,
                        \small$\Lambda^{(j)}_{2,\ket{0}\bra{0}}(n)$,
                        	\small{Bound for $\Lambda^{(j)}_{1,\ket{0}\bra{0}}(n)$},
                        	\small{Bound for $\Lambda^{(j)}_{1}(n)=\Lambda^{(j)}_{2}(n)$},
                        	\small{Bound for $\Lambda^{(j)}_{2,\ket{0}\bra{0}}(n)$}
                        }]
        	\addlegendimage{mark=*, only marks, color = blue}
        	\addlegendimage{mark=*, only marks, color = red}
        	\addlegendimage{mark=*, only marks, color = green}
        	\addlegendimage{mark=none, solid, color = blue} 
        	\addlegendimage{mark=none, solid, color = red} 
        	\addlegendimage{mark=none, solid, color = green}    	 
        \end{customlegend}
\end{tikzpicture}
\caption{Comparison of the bounds for random trajectory dynamical decoupling with a numerical simulation. The dots represent the numerical data whereas the lines are the bounds presented in Secs.~\ref{sec:Tr2} and~\ref{sec:Tr1}. The blue dots and lines are for the reduced evolution of system 1 when the initial state of system 2 is $\sigma_2=\ket{0}\bra{0}$. The red dots and lines are for the reduced evolution of system 1 (or 2) when the initial state of system 2 (or 1) is $\sigma_2=\mathbb{1}_2/d_2$ (or $\sigma_1=\mathbb{1}_1/d_1$). The green dots and lines portray the reduced evolution of system 2 in the case where system 1 starts in $\sigma_1=\ket{0}\bra{0}$. We have considered one qubit on system 1 and one qubit on system 2, i.e.~$d_1=d_2=2$. The Hamiltonian $H$, which generates the dynamics, is a generically chosen traceless Hermitian matrix. See Appendix~\ref{appendix:numerics} for details. The total evolution time $t$ is chosen such that $T=t\|\hat{\mathcal{H}}\|_\infty=1$. At each time step $t/n$, a unitary decoupling operation is drawn at random from $\mathscr{V}=\{\mathbb{1},X,Y,Z\}$, where $X$, $Y$, and $Z$ are the Pauli matrices.}
\label{fig:DD_full}
\end{figure}

The next proposition shows that the trajectory evolutions of system 2 far from the unitary $\rme^{-\rmi tH_2}$ become rare and almost all the trajectories of the reduced dynamics of system 2 get close to the unitary.
\begin{proposition}[Probability of the trajectories of system 2 far from unitary]
\label{prop:Tr2Tail}
The probability $\mathbb{P}(\le r)$ of finding a trajectory of the reduced dynamics of system 2 whose fidelity to the unitary $\rme^{-\rmi tH_2}$ measured by $\frac{1}{d_2}(\rme^{-\rmi tH_2}|\Lambda_{2,\sigma_1}^{(j)}|\rme^{-\rmi tH_2})$ is less than $r$ is bounded by
\begin{equation}
\mathbb{P}(\le r)
\le
\frac{1}{n}\frac{\sqrt{d_1}\,\|\sigma_1\|_2}{1-r}T^2,
\end{equation}
for any $t\ge0$ and $n\in\mathbb{N}$.
\end{proposition}
\begin{proof}
This is proved in Sec.~\ref{sec:proofTr2}.	
\end{proof}

Since the purity, the spectral norm, the fidelity, and the distance to the unitary are all bounded quantities, we can also bound their variances, instead of the expectation values, by making use of the Bhatia-Davis inequality $\Var[X]\leq(\max[X]-\E[X])(\E[X]-\min[X])$~\cite{BD2000}.
For instance, corresponding to Theorems~\ref{thm:Tr2P} and~\ref{thm:Tr2D}, we have the following corollary.
\begin{corollary}[Variances of the purity and of the distance from the Zeno dynamics of the trajectory evolution of system 2]
\label{cor:Tr2V}
\begin{gather}
\Var\!\left[
\mathtt{P}\bigl(\Lambda_{2,\sigma_1}^{(j)}(n)\bigr)
\right]
\le
(1-1/d_2)
\left[
1-\left(
1-\frac{1}{n}\sqrt{d_1}\,\|\sigma_1\|_2T^2
\right)^2
\right],
\\
\Var\Bigl[\|\Lambda_{2,\sigma_1}^{(j)}(n)\|_\infty\Bigr]
\le
\frac{1}{n}
(1-1/d_2)
\sqrt{d_1}\,\|\sigma_1\|_2T^2,
\\
\Var\!\left[
\frac{1}{d_2}
(\rme^{-\rmi tH_2}|\Lambda_{2,\sigma_1}^{(j)}(n)|\rme^{-\rmi tH_2})
\right]
\le
\frac{1}{n}\sqrt{d_1}\,\|\sigma_1\|_2T^2,
\\
\Var\!\left[
\left\|
\Lambda_{2,\sigma_1}^{(j)}(n)
-\frac{1}{d_2}|\rme^{-\rmi tH_2})(\rme^{-\rmi tH_2}|
\right\|_2
\right]
\le
2 \sqrt{\frac{2}{n}\sqrt{ {d_1}}\, \Fnorm{\sigma_1} T^2},
\\
\Var\Bigl[
\|
\mathcal{E}_{2,\sigma_1}^{(j)}
-
\rme^{-\rmi t\mathcal{H}_2}
\|_\diamond
\Bigr]
\le
2d_2^2\sqrt{\frac{2}{n}\sqrt{d_1}\,\|\sigma_1\|_2T^2},
\end{gather}
for any $t\ge0$ and $n\in\mathbb{N}$.
\end{corollary}
\noindent
These bounds on the variances all shrink as $\order{1/n}$ in the decoupling limit $n\to\infty$, except for the bounds on the variances of the distances to the unitary, which show $\order{1/\sqrt{n}}$ behaviors.

\subsection{Convergence of the Trajectory Evolution of System 1}\label{sec:Tr1}
Now, recall that each trajectory evolution of the total system $\mathcal{E}_n^{(j)}(t)$ is unitary and the corresponding Choi-Jamio\l{}kowski state $\Lambda^{(j)}(n)$ is pure.
Its reduced Choi-Jamio\l{}kowski states of systems 1 and 2,
\begin{align}
\Lambda_1^{(j)}(n)&=\tr_{22'}\Lambda^{(j)}(n)=\Lambda_{1,\mathbb{1}_2/d_2}^{(j)}(n),
\\
\Lambda_2^{(j)}(n)&=\tr_{11'}\Lambda^{(j)}(n)=\Lambda_{2,\mathbb{1}_1/d_1}^{(j)}(n),
\end{align}
share the same spectrum, which is determined by the Schmidt coefficients of the pure state $\Lambda^{(j)}(n)$.
Therefore, we have
\begin{align}
\|\Lambda_1^{(j)}(n)\|_\infty
&=\|\Lambda_2^{(j)}(n)\|_\infty,\label{eq:Schmidt_OpNorm}\\
\mathtt{P}\bigl(\Lambda_1^{(j)}(n)\bigr)
&=\mathtt{P}\bigl(\Lambda_2^{(j)}(n)\bigr).
\end{align}
The fact that the reduced Choi-Jamio\l{}kowski state $\Lambda_2^{(j)}(n)$ of system 2 approaches a pure state implies that the reduced Choi-Jamio\l{}kowski state $\Lambda_1^{(j)}(n)$ of system 1 also becomes pure with the same purity as $\Lambda_2^{(j)}(n)$.
By exploiting this fact, we prove the following two theorems for the trajectory evolution $\mathcal{E}_{1,\sigma_2}^{(j)}(n)$ of system 1 for a general initial state $\sigma_2$ of system 2, whose Choi-Jamio\l{}kowski state is denoted by $\Lambda_{1,\sigma_2}^{(j)}(n)$.
\begin{theorem}[Purity of the trajectory evolution of system 1]
\label{thm:Tr1P}
\begin{equation}
\sqrt{
\mathbb{E}\!\left[
\mathtt{P}\bigl(\Lambda_{1,\sigma_2}^{(j)}(n)\bigr)
\right]
}
\ge
\mathbb{E}\Bigl[\|\Lambda_{1,\sigma_2}^{(j)}(n)\|_\infty\Bigr]
\ge1-\frac{1}{n}d_2\|\sigma_2\|_\infty T^2,
\end{equation}
for any $t\ge0$ and $n\in\mathbb{N}$.
\end{theorem}
\begin{proof}
This is proved in Sec.~\ref{sec:proofTr1}.	
\end{proof}
\noindent
Both bounds on the purity and on the spectral norm of the Choi-Jamio\l{}kowski state converge to $1$ as $\order{1/n}$ in the decoupling limit $n\rightarrow\infty$. This scaling is also observed numerically. See the red and blue points in Figs.~\ref{fig:Purity} and~\ref{fig:OpNorm}.
\begin{theorem}[Distance of the trajectory evolution of system 1 from a pure Choi-Jamio\l{}kowski state]
\label{thm:Tr1D}
Let $|v_{1,\sigma_2}^{(j)}(n))$ be the normalized eigenvector belonging to the largest eigenvalue of $\Lambda_{1,\sigma_2}^{(j)}(n)$.
Then,
\begin{gather}
\mathbb{E}\!\left[
\left\|
\Lambda_{1,\sigma_2}^{(j)}(n)
-|v_{1,\sigma_2}^{(j)}(n))(v_{1,\sigma_2}^{(j)}(n)|
\right\|_\infty
\right]
\le
\frac{1}{n}d_2\|\sigma_2\|_\infty T^2,\label{eq:D1Choi}
\end{gather}
for any $t\ge0$ and $n\in\mathbb{N}$.
\end{theorem}
\begin{proof}
This is proved in Sec.~\ref{sec:proofTr1}.	
\end{proof}
\noindent
These bounds shrink as $\order{1/n}$ in the decoupling limit $n\rightarrow\infty$, which are consistent with numerical simulations as can be seen from the blue and red points in Fig.~\ref{fig:DU_Choi}. 
Again, the variances of the above quantities can be bounded by means of the Bhatia-Davis inequality~\cite{BD2000}.
\begin{corollary}[Variances of the purity and of the distance from a pure Choi-Jamio\l{}kowski state of the trajectory evolution of system 1]
\label{cor:Tr1V}
\begin{gather}
\Var\!\left[
\mathtt{P}\bigl(\Lambda_{1,\sigma_2}^{(j)}(n)\bigr)
\right]
\le
(1-1/d_1)
\left[
1-\left(
1-\frac{1}{n}d_2\|\sigma_2\|_\infty T^2
\right)^2
\right],
\\
\Var\Bigl[\|\Lambda_{1,\sigma_2}^{(j)}(n)\|_\infty\Bigr]
\le
\frac{1}{n}(1-1/d_1)
d_2\|\sigma_2\|_\infty T^2,
\\
\Var\!\left[
\left\|
\Lambda_{1,\sigma_2}^{(j)}(n)
-|v_{1,\sigma_2}^{(j)}(n))(v_{1,\sigma_2}^{(j)}(n)|
\right\|_\infty
\right]
\le
\frac{1}{n}d_2\|\sigma_2\|_\infty T^2,
\end{gather}
for any $t\ge0$ and $n\in\mathbb{N}$.
\end{corollary}

As summarized above, we prove the convergences of the reduced dynamics of systems 1 and 2 to unitaries in terms of the Choi-Jamio\l{}kowski states.
We are also able to bound the distances of the reduced dynamics from unitaries, once we manage to bound the purities of their Choi-Jamio\l{}kowski states as well as their operator norms.
We provide an interesting proposition that allows us to bound the distance of a map from a unitary by the purity and operator norm of the Choi-Jamio\l{}kowski state of the map. It is based on the idea of the leading Kraus approximation~\cite{CarignanDugas2019}.
\begin{proposition}[Distance of a channel from a unitary]
\label{thm:dist_unitary_purity}
Let $\mathcal{T}$ be a CPTP map acting on a $d$-dimensional quantum system and let $\Lambda$ be the Choi-Jamio\l{}kowski state of $\mathcal{T}$.
Define the map $\mathcal{U}=U{}\bullet{}U^\dag$ with $U$ the closest unitary to the Kraus operator $E_\mathrm{max}$ belonging to the largest eigenvalue of $\Lambda$ according to~(\ref{Kraus-choice}).
Then, the distance between the matrix representations $\hat{\mathcal{T}}$ and $\hat{\mathcal{U}}$ of $\mathcal{T}$ and $\mathcal{U}$ is bounded by
\begin{equation}
d(1-\sqrt{\pur})
\leq\|\hat{\mathcal{T}}-\hat{\mathcal{U}}\|_2
\leq d\sqrt{\pur-\pur^2} + d\sqrt{1-\pur^2}
\leq2d\sqrt{1-\pur^2},\label{eq:dist_unitary_purity}
\end{equation}
where $\pur=\pur(\Lambda)$ is the purity of $\Lambda$. Furthermore, the diamond distance between $\mathcal{T}$ and $\mathcal{U}$ is bounded by
\begin{equation}
	\|\mathcal{T}-\mathcal{U}\|_\diamond
	\le
	3d(1-\|\Lambda\|_\infty).\label{eq:diamond_dist_choi}
\end{equation}
\end{proposition}
\begin{proof}
This is proved in Sec.~\ref{sec:proofTr1}.	
\end{proof}
\noindent
A bound on the distance of the trajectory evolution of system 1 from a unitary can also be inferred from Theorem~\ref{thm:Tr1P} using Proposition~\ref{thm:dist_unitary_purity}. The Bhatia-Davis inequality~\cite{BD2000} gives a bound on the variance accordingly. For this, notice that the maximum Frobenius norm distance between two matrix representations of CPTP maps acting on a $d$-dimensional quantum system can be bounded by $2d$. This follows from the triangle inequality together with~\eqref{eq:Fnorm_purity}. Furthermore, the maximum diamond norm distance is bounded by $2$, which is inherited from the trace norm distance.
\begin{corollary}[Average and variance of the distance of the trajectory evolution of system 1 from a unitary]
\label{cor:Tr1D}
Let $\hat{\mathcal{E}}_{1,\sigma_2}^{(j)}(n)$ be the matrix representation of the reduced random trajectory dynamical decoupling evolution $\mathcal{E}_{1,\sigma_2}^{(j)}(n)$ on system 1. As in Proposition~\ref{thm:dist_unitary_purity}, define the map $\mathcal{U}_{1,\sigma_2}^{(j)}(n)=U_{1,\sigma_2}^{(j)}(n){}\bullet{}U_{1,\sigma_2}^{(j)\dag}(n)$ with $U_{1,\sigma_2}^{(j)}(n)$ the closest unitary to the Kraus operator $E_\mathrm{max}$ belonging to the largest eigenvalue of $\Lambda_{1,\sigma_2}^{(j)}(n)$ according to~(\ref{Kraus-choice}). Then,
	\begin{align}
		\mathbb{E}\Bigl[\Fnorm{\hat{\mathcal{E}}_{1,\sigma_2}^{(j)}(n)-\hat{\mathcal{U}}_{1,\sigma_2}^{(j)}(n)}\Bigr]
		&\leq
		2d_1\sqrt{1-\left(1-\frac{1}{n}d_2\infnorm{\sigma_2}T^2\right)^4},\label{eq:bound_matrix_rep}\\
		\Var\Bigl[\Fnorm{\hat{\mathcal{E}}_{1,\sigma_2}^{(j)}(n)-\hat{\mathcal{U}}_{1,\sigma_2}^{(j)}(n)}\Bigr]
		&\leq
		4d_1^2\sqrt{1-\left(1-\frac{1}{n}d_2\infnorm{\sigma_2}T^2\right)^4},\\
		\mathbb{E}\Bigl[\|\mathcal{E}_{1,\sigma_2}^{(j)}(n)-\mathcal{U}_{1,\sigma_2}^{(j)}(n)\|_\diamond\Bigr]
		&\leq
		\frac{3}{n}d_1d_2\|\sigma_2\|_\infty T^2,\label{eq:Diamond1}\\
		\Var\Bigl[\|\mathcal{E}_{1,\sigma_2}^{(j)}(n)-\mathcal{U}_{1,\sigma_2}^{(j)}(n)\|_\diamond\Bigr]
		&\leq
		\frac{6}{n}d_1d_2\|\sigma_2\|_\infty T^2,
	\end{align}
where $\hat{\mathcal{U}}_{1,\sigma_2}^{(j)}(n)$ is the matrix representation of $\mathcal{U}_{1,\sigma_2}^{(j)}(n)$.
\end{corollary}
\noindent
The bound in~\eqref{eq:bound_matrix_rep} on the distance from the unitary in the matrix representation diminishes in the decoupling limit $n\rightarrow\infty$ as $\order{1/\sqrt{n}}$. Numerically, we observe $\order{1/n}$. See the blue and red points in Fig.~\ref{fig:DU_Super}. For the bound on the diamond distance in~\eqref{eq:Diamond1}, we find a $\order{1/n}$ scaling.

In order to gain some intuition about the random trajectories, an example of the purity for a typical random trajectory as well as an example of the purity for an atypical random trajectory are shown in Fig~\ref{fig:typical_atypical_trajectory}. The purity of the atypical trajectory behaves very differently from the average purity whereas the purity of the typical trajectory exhibits a similar qualitative nature. This circumstance is a consequence of the fact that almost all random trajectories lead to system evolutions close to a unitary in the decoupling limit $n\rightarrow\infty$. This can be seen by looking at the probability of getting ``bad'' trajectories. In Fig.~\ref{fig:Probability_bad_purity}, the estimated probability of achieving only low purity of the reduced Choi-Jamio\l{}kowsi state ($\pur\le 0.99$) is shown as a function of the number $n$ of decoupling steps. The numerical experiment reveals an exponential decay of the probability of ``bad'' trajectories with increasing $n$.
\begin{figure}
\centering
	\begin{subfigure}[r]{.45\textwidth}
  	\begin{tikzpicture}[mark size={0.6}, scale=1]
	\begin{axis}[
	xmode=log,
	ymode=log,
	xlabel={$n$},
	ylabel={$1-\pur\big(\Lambda_{1,\ket{0}\bra{0}}\big)$},
	x post scale=0.8,
	y post scale=0.8,
	legend pos=south west,
	legend cell align={left},
	]
	\addplot[color=blue, only marks] table[x=n, y=P, col sep=comma] {Purity_1_sigma.csv};
	\addlegendentry{\small $\mathbb{E}[\pur(\Lambda_{1,\ket{0}\bra{0}}^{(j)})]$};
	\addplot[color=red, only marks] table[x=n, y=P, col sep=comma] {Purity_atyp_1_sigma.csv};
	\addlegendentry{\small $\pur(\Lambda_{1,\ket{0}\bra{0}}^\mathrm{atypical})$};
	\addplot[color=green, only marks] table[x=n, y=P, col sep=comma] {Purity_typ_1_sigma.csv};
	\addlegendentry{\small $\pur(\Lambda_{1,\ket{0}\bra{0}}^\mathrm{typical})$};
	\end{axis}
	\end{tikzpicture}
  	\caption{Purity of the reduced Choi-Jamio\l{}kowski state $\Lambda_{1,\ket{0}\bra{0}}^{(j)}$ of the trajectory evolution under random dynamical decoupling for the initial state $\sigma_2=\ket{0}\bra{0}$ of system 2. Blue: Average over 100 random realizations. Red: random trajectory through an atypical decoupling sequence. Green: random trajectory through a typical decoupling sequence.}
  	\label{fig:typical_atypical_trajectory}
	\end{subfigure}
	\hfill
	\begin{subfigure}[l]{.45\textwidth}
  	\begin{tikzpicture}[mark size={0.6}, scale=1]
	\begin{axis}[
	ymode=log,
	ylabel near ticks,
	yticklabel pos=right,
	xlabel={$n$},
	ylabel={$\mathbb{P}\Big[\pur\big(\Lambda_{1,\ket{0}\bra{0}}^{(j)}\big)\le 0.99\Big]$},
	x post scale=0.8,
	y post scale=0.8,
	]
	\addplot[color=blue, only marks] table[x=n, y=P, col sep=comma] {ProbPur1s.csv};
	\end{axis}
	\end{tikzpicture} 
  	\caption{Estimated Probability of obtaining $\pur(\Lambda_{1,\ket{0}\bra{0}}^{(j)})\le 0.99$. For each decoupling step $n$, 1000 random sequences are sampled, and the number of trajectories leading to Choi-Jamio\l{}kowski states of purity less than or equal to 0.99 is counted. The plot shows that the probability of such ``bad'' random trajectories decreases exponentially with $n$.}
  	\label{fig:Probability_bad_purity}
	\end{subfigure}
	\begin{tikzpicture}
\end{tikzpicture}
\caption{
Purities through typical and atypical decoupling sequences. 
The model is exactly the same as the one in Fig.~\ref{fig:DD_full}, i.e.\ one qubit on system 1 and one qubit on system 2. The Hamiltonian $H$ is a generically chosen traceless Hermitian matrix and the total evolution time is $t=1/\|\mathcal{\hat{H}}\|_\infty$. See Appendix~\ref{appendix:numerics} for details. System 2 is initialized in the state $\sigma_2=\ket{0}\bra{0}$. In subfigure (a), typical and atypical sample trajectories are compared with the average over 100 different random trajectories. The typical trajectory behaves similarly to the average whereas the atypical trajectory does not even show decoupling. This indicates that almost all trajectories will lead to a system evolution close to unitary. The typical trajectory is just one particular random trajectory through the decoupling group $\mathscr{V}=\{\mathbb{1},X,Y,Y\}$, where each element is chosen uniformly at random. The atypical trajectory is obtained by changing the probability distribution and making the drawing of the element $\mathbb{1}\in\mathscr{V}$ 20 times more likely than all the others. Both trajectories obtained by this procedure are concretely specified in Appendix~\ref{appendix:numerics}\@. In subfigure (b), the estimated probability of getting ``bad'' trajectories is shown as a function of the number of decoupling steps $n$. We observe an exponential decay of this probability, which further clarifies that almost all random trajectories yield system evolutions close to unitary in the decoupling limit $n\rightarrow\infty$.}
\label{fig:probability_full}
\end{figure}
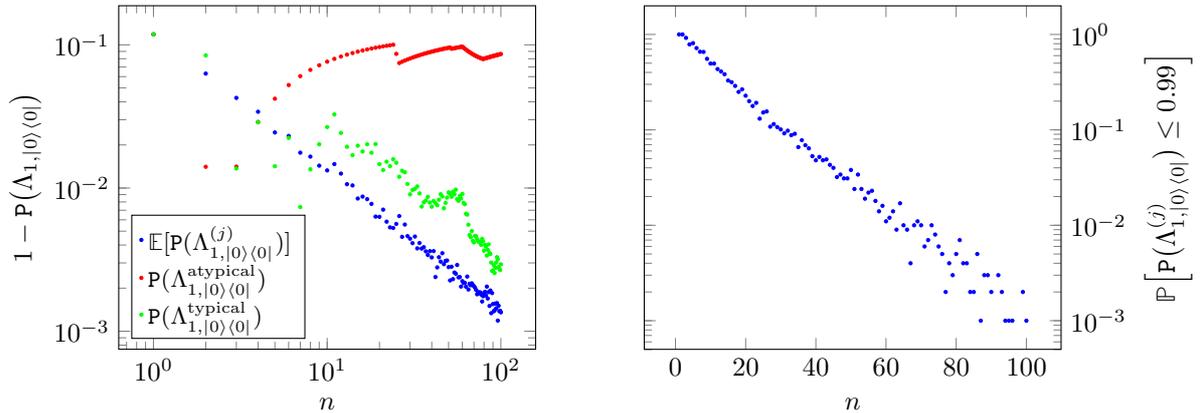

The bounds in Corollary~\ref{cor:Tr1D} specify the distance between the evolution of system 1 and some unknown unitary $\mathcal{U}_{1,\sigma_2}^{(j)}$ under decoupling sequence. If one keeps track of the applied random pulses, one could apply the inverse of the product of all the pulses at the end of the decoupling sequence, i.e.,
\begin{equation}
	\mathcal{\tilde{E}}_n^{(j)}(t)=(\mathcal{V}_{n+1}^{(j)}\cdots\mathcal{V}_1^{(j)})^\dag\mathcal{E}_n^{(j)}(t).\label{eq:pulse_inversion}
\end{equation}
In this case, one would expect that the resulting reduced evolution $\mathcal{\tilde{E}}_{1,\sigma_2}^{(j)}(n)$ of system 1 is close to the identity $\mathbb{I}_1$. Indeed, we numerically find that $\mathcal{\tilde{E}}_{1,\sigma_2}^{(j)}(n)$ converges to $\mathbb{I}_1$ in the decoupling limit $n\rightarrow\infty$. See the blue dots in Fig.~\ref{fig:pulse_inversion}. Notice that the distances plotted in Fig.~\ref{fig:pulse_inversion} in terms of the Choi-Jamio\l{}kowski states are equivalent to the corresponding diamond distances due to Lemma~\ref{lemma:diamond_choi}. Therefore, the green dots in Fig.~\ref{fig:pulse_inversion} shows that the bound in~\eqref{eq:Diamond1} captures the correct asymptotic scaling in $n$, i.e., $\mathcal{E}_{1,\sigma_2}^{(j)}(n)$ converges to $\mathcal{U}_{1,\sigma_2}^{(j)}(n)$ in diamond norm as $\order{1/n}$. However, we numerically find that the reduced evolution $\tilde{\mathcal{E}}_{1,\sigma_2}^{(j)}(n)$ after undoing the random pulses converges to the identity $\mathbb{I}_1$ only as $\order{1/\sqrt{n}}$, as shown by the blue dots in Fig.~\ref{fig:pulse_inversion}. This difference in scaling manifests in the numerical evidence that $\tilde{\mathcal{U}}_{1,\sigma_2}^{(j)}(n)$ also converges to the identity $\mathbb{I}_1$ only as $\order{1/\sqrt{n}}$. See the red dots in Fig.~\ref{fig:pulse_inversion}. This makes sense as by the triangle inequality and the unitary invariance of the trace norm we have
\begin{align}
	\left\|\tilde{\Lambda}_{1,\sigma_2}^{(j)}(n)-\frac{1}{d_1}|\mathbb{1}_{11'})(\mathbb{1}_{11'}|\right\|_1
	\le{}&
	\left\|
	\Lambda_{1,\sigma_2}^{(j)}(n)
	-\lambda_0|v_{1,\sigma_2}^{(j)}(n))(v_{1,\sigma_2}^{(j)}(n)|
	\right\|_1\nonumber\\
	&{}+
	\left\|\lambda_0|\tilde{v}_{1,\sigma_2}^{(j)}(n))(\tilde{v}_{1,\sigma_2}^{(j)}(n)|
	-\frac{1}{d_1}|\mathbb{1}_{11'})(\mathbb{1}_{11'}|
	\right\|_1,
	\label{eq:triangle_pulse_inv}
\end{align}
where $|v_{1,\sigma_2}^{(j)}(n))$ and $|\tilde{v}_{1,\sigma_2}^{(j)}(n))$ are the eigenvectors of $\Lambda_{1,\sigma_2}^{(j)}(n)$ and $\tilde{\Lambda}_{1,\sigma_2}^{(j)}(n)$, respectively, belonging to their largest eigenvalue $\lambda_0$.
Unfortunately, it is not possible by our method to give explicit bounds on the quantities including the pulse inversions. This is because we do not have enough information to relate $\mathcal{\tilde{U}}_{1,\sigma_2}^{(j)}(n)$ to $\mathbb{I}_1$, or equivalently, $\mathcal{U}_{1,\sigma_2}^{(j)}(n)$ to the product of all the applied random pulses, which can be seen by the following simple argument. Consider the extreme case where there is no interaction between systems 1 and 2, and we only aim to completely switch off the Hamiltonian of system 1 by the randomly chosen pulses. In this case, the target evolution of system 1 is the identity map $\mathbb{I}_1$ and the efficiency of the random dynamical decoupling is given by the distance of the reduced evolution of system 1 to the identity $\|\tilde{\mathcal{E}}_1^{(j)}(n)-\mathbb{I}_1\|_\diamond$. However, since there is no interaction between systems 1 and 2, the reduced evolution $\mathcal{E}_1^{(j)}(n)$ of system 1 is for sure unitary, no matter how good or bad the decoupling works. Therefore, the distance of the reduced evolution $\mathcal{E}_1^{(j)}(n)$ to the unitary $\mathcal{U}_1^{(j)}(n)$ specified by the Kraus operator corresponding to the largest eigenvalue of its Choi-Jamio\l{}kowski state is always zero, $\|\mathcal{E}_1^{(j)}(n)-\mathcal{U}_1^{(j)}(n)\|_\diamond=0$. Nevertheless, $\mathcal{U}^{(j)}(n)$ can be far away from the product of all the applied pulses and the distance $\|\tilde{\mathcal{E}}_1^{(j)}(n)-\mathbb{I}_1\|_\diamond$ can be big. In this way, the distance to a unitary is not very informative about the convergence to the target unitary.

\begin{figure}
\centering
	\begin{tikzpicture}[mark size={1}, scale=1]
	\begin{axis}[
	xmode=log,
	ymode=log,
	xlabel={$n$},
	ylabel={$\mathbb{E}\left[\|\Lambda(\mathcal{S})-\Lambda(\mathcal{T})\|_1\right]$},
	x post scale=1.2,
	y post scale=1.2,
	transpose legend,
	legend columns = 3,
	legend pos = south west,
	legend cell align={left},
	]
	\addplot[color=blue, only marks] table[x=n, y=DU, col sep=comma]{DU_Choi_inv_1_map_id_sigma.csv};
	\addplot[color=red, only marks] table[x=n, y=DU, col sep=comma]{DU_Choi_inv_U_id_1_sigma.csv};
	\addplot[color=green, only marks] table[x=n, y=DU, col sep=comma]{DU_Choi_1_sigma_trnorm.csv};
	\legend{$\mathcal{S}=\tilde{\mathcal{E}}^{(j)}_{1,\ket{0}\bra{0}}\text{, }\mathcal{T}=\mathbb{I}_1$, $\mathcal{S}=\tilde{\mathcal{U}}^{(j)}_{1,\ket{0}\bra{0}}\text{, }\mathcal{T}=\mathbb{I}_1$, $\mathcal{S}=\mathcal{E}^{(j)}_{1,\ket{0}\bra{0}}\text{, }\mathcal{T}=\mathcal{U}^{(j)}_{1,\ket{0}\bra{0}}$}
	\end{axis}
	\end{tikzpicture}
\caption{
Numerical evaluation of the quantities in~\eqref{eq:triangle_pulse_inv}. The maps with tildes represent the evolutions with all the random pulses being inverted at the end of the decoupling sequence. See~\eqref{eq:pulse_inversion}. In this case, the reduced evolution $\mathcal{\tilde{E}}_{1,\sigma_2}^{(j)}$ of system 1 converges to the identity map $\mathbb{I}_1$ as $\mathcal{O}(1/\sqrt{n})$, which is shown by the blue dots. Similarly, the closest unitary $\mathcal{\tilde{U}}_{1,\sigma_2}^{(j)}$ to the Kraus operator corresponding to the largest eigenvalue of the Choi-Jamio\l{}kowski state $\tilde{\Lambda}_{1,\sigma_2}^{(j)}$ of $\mathcal{\tilde{E}}_{1,\sigma_2}^{(j)}$ converges to the identity map $\mathbb{I}_1$ as $\mathcal{O}(1/\sqrt{n})$, as can be seen from the red dots. In contrast, the distance between $\mathcal{E}_{1,\sigma_2}^{(j)}$ and $\mathcal{U}_{1,\sigma_2}^{(j)}$ shrinks as $\mathcal{O}(1/n)$. It is depicted with the green dots. Note that this distance is the same as the distance between $\mathcal{\tilde{E}}_{1,\sigma_2}^{(j)}$ and $\mathcal{\tilde{U}}_{1,\sigma_2}^{(j)}$ by the unitary equivalence of the trace norm. By the triangle inequality, $\text{blue}\le\text{red}+\text{green}$. See~\eqref{eq:triangle_pulse_inv}. The distances plotted here, in terms of the trace norm of the Choi-Jamio\l{}kowski state, are equivalent to the diamond distances of the corresponding channels by Lemma~\ref{lemma:diamond_choi}. For this numerical evaluation, we used the same model as in Fig.~\ref{fig:DD_full}, i.e.\ one qubit on system 1 and one qubit on system 2 with a generically chosen Hamiltonian $\mathcal{H}$, and the total evolution time is $t=1/\|\mathcal{\hat{H}}\|_\infty$. See Appendix~\ref{appendix:numerics} for details. The initial state of system 2 is chosen to be $\sigma_2=\ket{0}\bra{0}$. In order to estimate the expectation value, we sampled $100$ random realizations.
}
\label{fig:pulse_inversion}
\end{figure}
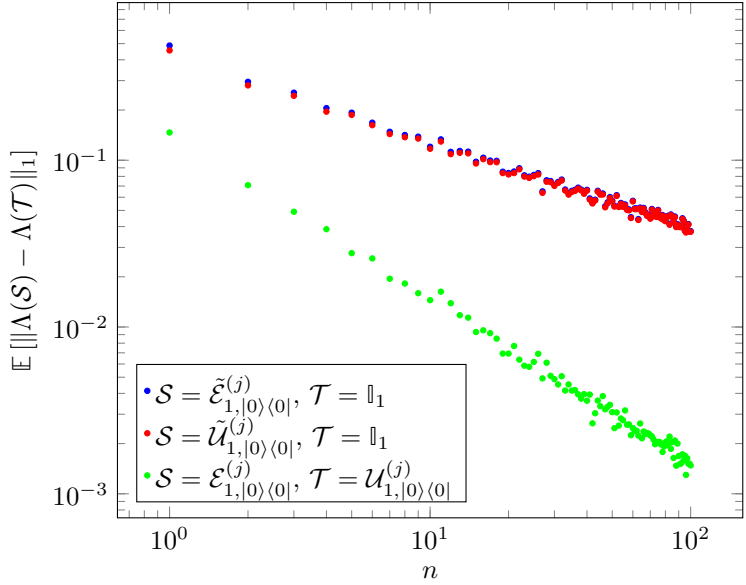

In summary, all decoupling bounds presented here are ruled by the Zeno bound derived in Theorem~\ref{thm:convergence_average}. This shows a deep connection between the quantum Zeno effect and the random dynamical decoupling through a phenomenon that we call \emph{equitability of system and bath}. This is, by using the fact that the bath evolution (system 2) is close to unitary, we infer that the system evolution (system 1) is so as well. Hence, there is no preferred subsystem for quantum control and both $\mathscr{H}_1$ and $\mathscr{H}_2$ can be treated equitable. Physically, this idea is based on the fact that we only need to remove the interaction term in the Hamiltonian to achieve decoupling. If that happens \emph{both} subsystems evolve unitarily. Mathematically, the approach is justified by the fact that the total evolution $\mathcal{E}_n^{(j)}(t)$ is unitary even in the presence of controls. Thus, its Choi-Jamio\l{}kowski state $\Lambda^{(j)}(n)$ is pure. Due to the Schmidt-decomposition, the reduced Choi-Jamio\l{}kowski states then have the same purity. As a consequence, it suffices to enhance the purity of one subsystem to obtain a higher purity on \emph{both} subsystems. On one hand, performing projective measurements on system 1 leads to a unitary Zeno dynamics on system 2. On the other hand, random dynamical decoupling gives rise to a unitary evolution on system 1. However, the idea of equitability of system and bath shows that these two are essentially the same.

The mathematical structure of this approach to the random dynamical decoupling might remind some readers of the procedure of randomized benchmarking. Therefore, it is worth briefly discussing similarities and differences between the two. On one hand, the projection $\mathcal{D}$ that specifies the average random dynamical decoupling evolution in the Zeno limit projects $\mathscr{H}_1$ to the group average of the decoupling group $\mathscr{V}$. Due to the irreducibility assumption, $\mathcal{D}$ is a $\mathscr{V}$-twirl over a unitary quantum $1$-design on $\mathscr{H}_1$, which will average away (switch off) the noise induced by the bath. On the other hand, in randomized benchmarking, one applies random Clifford operations to a quantum system in the presence of noise, i.e., one implements a twirl over a unitary quantum $2$-design. By this procedure, one averages over different noise sources to make the noise more isotropic or less correlated. Whilst there is this apparent mathematical similarity between random dynamical decoupling and randomized benchmarking of twirling the noise over a unitary quantum $t$-design, these two methods are different from a physical perspective. The goal of randomized benchmarking is to characterize the noise whereas random dynamical decoupling aims to remove it. Furthermore, the description of random dynamical decoupling relies on the full model of a quantum system together with a quantum bath. In contrast, in randomized benchmarking, the bath is not specified and only a noisy quantum system is considered. A study of explicit error bounds for the average gate fidelities in the context of randomized benchmarking can be found in Ref.~\cite{Wallman_2014}.

\section{Convergence of the Quantum Zeno Limit and the Average Dynamical Decoupling}\label{sec:Zeno_av}
Ultimately, our goal is to show that the decoupling error of random dynamical decoupling is essentially determined by the convergence of the quantum Zeno limit, and hence to prove Theorems~\ref{thm:Tr1P} and~\ref{thm:Tr1D}. The route we take is to infer information from the average evolution. This section is dedicated to the average protocol, hence we prove the statements of Sec.~\ref{sec:Av2}. The evolution under the average protocol can be understood as a special case of quantum Zeno dynamics. A very general treatment of the quantum Zeno dynamics by arbitrary quantum operations has been developed in Ref.~\cite{Burgarth2018}. Here, we only focus on the standard procedure with repeated projective measurements. We first discuss the convergence of this scheme and prove a bound, which scales as $\order{1/n}$ (Theorem~\ref{thm:convergence_average}). Afterwards, we show how this relates to the average dynamical decoupling protocol. This discussion directly gives us a convergence bound on the average evolution under the random dynamical decoupling (Propositions~\ref{prop:Av2D} and~\ref{prop:Av2P}).

\subsection{Convergence of the Quantum Zeno Limit}\label{sec:Zeno}
In this subsection, we discuss the convergence of the quantum Zeno limit and prove Theorem~\ref{thm:convergence_average}. We will extensively make use of the fact that unitaries as well as Hermitian projection operators have norm $1$ in the operator norm, e.g.~$\infnorm{\rme^{-\rmi tH}}=\infnorm{P}=\infnorm{\mathbb{1}-P}=1$. Let us start by recalling an auxiliary lemma.
\begin{lemma}\label{lemma:aux_lemma_av}
For any Hermitian operator $X=X^\dagger$ and for any $t\geq 0$, we have
\begin{align}
\|\rme^{-\rmi tX}-\mathbb{1}\|_\infty
&\leq t\|X\|_\infty,
\label{eq:aux_lemma_1}\\
\|\rme^{-\rmi tX}-\mathbb{1}+\rmi tX\|_\infty
&\leq\frac{1}{2}t^2\|X\|_\infty^2.
\label{eq:aux_lemma_2}
\end{align}
\end{lemma}
\begin{proof}
Observe first that
\begin{equation}
\rme^{-\rmi tX}-\mathbb{1}
=\int_0^t\rmd s\,\frac{\partial}{\partial s}\rme^{-\rmi sX}
=-\rmi\int_0^t\rmd s\,X\rme^{-\rmi sX},
\label{eq:e-1}
\end{equation}
which yields the first bound (\ref{eq:aux_lemma_1}).
Continuing (\ref{eq:e-1}) by integration by parts,
\begin{align}
\rme^{-\rmi tX}-\mathbb{1}
&=
\rmi(t-s)X\rme^{-\rmi sX}\biggr|_{s=0}^{s=t}
-\int_0^t\rmd s\,(t-s)X^2\rme^{-\rmi sX}
\nonumber\\
&=-\rmi tX
-\int_0^t\rmd s\,(t-s)X^2\rme^{-\rmi sX},
\label{eq:X=H_X=PHP}
\end{align}
this gives the second bound (\ref{eq:aux_lemma_2}).
\end{proof}
\noindent
We apply Lemma~\ref{lemma:aux_lemma_av} to prove a bound on the following quantum Zeno limit.
\begin{proposition}\label{lemma:Dav_convergence}
Let $H=H^\dagger$ be a Hermitian operator and let $P=P^2=P^\dagger$ be a Hermitian projection. 
Then, for any $t\ge0$ and $n\in\mathbb{N}$, we have
\begin{equation}
\|(P\rme^{-\rmi \frac{t}{n}H}P)^n-\rme^{-\rmi tPHP}P\|_\infty
\le\frac{t^2}{n}\|H\|_\infty^2.
\label{eq:Dav_prop}
\end{equation}
\end{proposition}
\begin{proof}
First, observe that
\begin{align}
\|
P\rme^{-\rmi\frac{t}{n}H}P
-\rme^{-\rmi\frac{t}{n}PHP}P
\|_\infty
&=\left\|
P
\left(
\rme^{-\rmi\frac{t}{n}H}-\mathbb{1}+\rmi\frac{t}{n}H
\right)
P
-
\left(
\rme^{-\rmi\frac{t}{n}PHP}
-\mathbb{1}+\rmi\frac{t}{n}PHP
\right)
P
\right\|_\infty
\nonumber\\
&\leq
\left\|
\rme^{-\rmi\frac{t}{n}H}
-\mathbb{1}
+\rmi\frac{t}{n}H
\right\|_\infty
+
\left\|
\rme^{-\rmi\frac{t}{n}PHP}
-\mathbb{1}+\rmi\frac{t}{n}PHP
\right\|_\infty
\nonumber\\
&\leq
\frac{t^2}{2n^2}\|H\|_\infty^2
+
\frac{t^2}{2n^2}\|PHP\|_\infty^2
\nonumber\\
&\leq
\frac{t^2}{n^2}\|H\|_\infty^2,
\label{eq:step2intermediate}
\end{align}
where we have used Lemma~\ref{lemma:aux_lemma_av} for the second inequality.
We then use a standard telescope sum trick
\begin{equation}
A^{n}-B^{n}=\sum_{k=0}^{n-1}A^{k}(A-B)B^{n-1-k},
\label{eq:Telescope_trick}
\end{equation}
for any operators $A$ and $B$ for any $n\in\mathbb{N}$. This is proved in Lemma~\ref{lemma:telescope} in Appendix~\ref{appendix:lemmas}\@.
By applying \eqref{eq:Telescope_trick} to $A=P\rme^{-\rmi\frac{t}{n}H}P$
		and $B=\rme^{-\rmi\frac{t}{n}PHP}P$ and using \eqref{eq:step2intermediate}, we get
\begin{align}
\|
(P\rme^{-\rmi\frac{t}{n}H}P)^n
-
\rme^{-\rmi tPHP}P
\|_\infty
&\leq
\sum_{k=0}^{n-1}
\|P\rme^{-\rmi\frac{t}{n}H}P\|_\infty^k
\|P\rme^{-\rmi\frac{t}{n}H}P-\rme^{-\rmi tPHP}P\|_\infty
\|\rme^{-\rmi tPHP}P\|_\infty^{n-1-k}
\nonumber\\
&\le
\frac{t^2}{n}\infnorm{H}^2,
\end{align}
which proves the proposition.
\end{proof}
\noindent
To prove Theorem~\ref{thm:convergence_average} presented in Sec.~\ref{sec:Av2}, we need to adjust the end of the sequence of repeated projections with the help of the following lemma.
\begin{lemma}
\label{lem:GoldenLemmaH}
Let $H=H^\dag$ be a Hermitian operator and let $P=P^2=P^\dag$ be a Hermitianity projection.
Then, for any $t\ge0$ and $n\in\mathbb{N}$, we have
\begin{equation}
\|
(\rme^{-\rmi\frac{t}{n}H}P)^n
-
(P\rme^{-\rmi\frac{t}{n}H}P)^n
\|_\infty
\le
\frac{t}{n}
\|H\|_\infty.
\label{eqn:GoldenLemmaH}
\end{equation}
\end{lemma}
\begin{proof}
We first split as
\begin{align}
(\rme^{-\rmi\frac{t}{n}H}P)^n
&=
P(\rme^{-\rmi\frac{t}{n}H}P)^n
+
(\mathbb{1}-P)(\rme^{-\rmi\frac{t}{n}H}P)^n
\nonumber\\
&=
(P\rme^{-\rmi\frac{t}{n}H}P)^n
+
(\mathbb{1}-P)\rme^{-\rmi\frac{t}{n}H}P(\rme^{-\rmi\frac{t}{n}H}P)^{n-1}
\nonumber\\
&=
(P\rme^{-\rmi\frac{t}{n}H}P)^n
+
(\mathbb{1}-P)(\rme^{-\rmi\frac{t}{n}H}-\mathbb{1})P(\rme^{-\rmi\frac{t}{n}H}P)^{n-1}.
\end{align}
Therefore,
\begin{align}
\|
(\rme^{-\rmi\frac{t}{n}H}P)^n
-
(P\rme^{-\rmi\frac{t}{n}H}P)^n
\|_\infty
&=
\|
(\mathbb{1}-P)(\rme^{-\rmi\frac{t}{n}H}-\mathbb{1})P(\rme^{-\rmi\frac{t}{n}H}P)^{n-1}
\|_\infty
\vphantom{\frac{t}{n}}
\nonumber\\
&\le
\|
\rme^{-\rmi\frac{t}{n}H}-\mathbb{1}
\|_\infty
\vphantom{\frac{t}{n}}
\nonumber\\
&\le
\frac{t}{n}\|H\|_\infty,
\label{eqn:GoldenLemmaH}
\end{align}
where we have used Lemma~\ref{lemma:aux_lemma_av}.
\end{proof}
\noindent
Theorem~\ref{thm:convergence_average} in Sec.~\ref{sec:Av2} then follows by combining Proposition~\ref{lemma:Dav_convergence} and Lemma~\ref{lem:GoldenLemmaH} through the triangle inequality. 
\begin{proof}[Proof of Theorem~\ref{thm:convergence_average}]
Using Proposition~\ref{lemma:Dav_convergence} and Lemma~\ref{lem:GoldenLemmaH},
\begin{align}
\|(\rme^{-\rmi \frac{t}{n}H}P)^n-\rme^{-\rmi tPHP}P\|_\infty
&\le
\|(\rme^{-\rmi \frac{t}{n}H}P)^n-(P\rme^{-\rmi \frac{t}{n}H}P)^n\|_\infty
+
\|(P\rme^{-\rmi \frac{t}{n}H}P)^n-\rme^{-\rmi tPHP}P\|_\infty
\nonumber\\
&\le\frac{t}{n}\|H\|_\infty + \frac{t^2}{n}\|H\|_\infty^2.
\end{align}
\end{proof}
\begin{remark}
	The dynamics considered here starts with a projective measurement followed by a free evolution for a short time $t/n$. This procedure is repeated $n$ times. Alternatively, one could start with the free evolution and perform projective measurements $n$ times at regular time intervals $t/n$ during the evolution. Notice that this only makes a difference in the first and the last step. Our bound is still valid in this case, i.e.,
	\begin{equation}
		\|(P\rme^{-\rmi \frac{t}{n}H})^n-\rme^{-\rmi tPHP}P\|_\infty
		\le\frac{t}{n}\|H\|_\infty + \frac{t^2}{n}\|H\|_\infty^2.
	\end{equation}
	To see this, we only need to change Lemma~\ref{lem:GoldenLemmaH} accordingly.
\end{remark}

We now turn our attention to the average dynamical decoupling protocol, which is a variant of the quantum Zeno dynamics.

\subsection{Convergence of the Average Dynamical Decoupling}\label{sec:av}
As per~\eqref{eq:average_protocol}, the matrix representation of the average evolution under the random dynamical decoupling protocol takes the form
\begin{equation}
\hat{\mathcal{E}}_n^\mathrm{av}(t)
=(\hat{\mathcal{D}}\rme^{-\rmi\frac{t}{n}\hat{\mathcal{H}}}\hat{\mathcal{D}})^n,
\end{equation}
where $\hat{\mathcal{D}}=\frac{1}{d_1}\vecket{\mathbb{1}_1}\vecbra{\mathbb{1}_1}\otimes\hat{\mathbb{I}}_{22'}$ is a Hermitian projection. Therefore, Proposition~\ref{lemma:Dav_convergence} can be applied to obtain a bound on its convergence,
\begin{equation}
	\|
	(\hat{\mathcal{D}}\rme^{-\rmi \frac{t}{n}\hat{\mathcal{H}}}\hat{\mathcal{D}})^n
	-
	\rme^{-\rmi t\hat{\mathcal{D}} \hat{\mathcal{H}}\hat{\mathcal{D}}}\hat{\mathcal{D}}
	\|_\infty
	\leq
	\frac{t^2}{n}\|\hat{\mathcal{H}}\|_\infty^2
	.\label{eq:av_error}
\end{equation}
Now, we see why this procedure is called dynamical decoupling. In the Zeno Hamiltonian $\hat{\mathcal{H}}_Z=\hat{\mathcal{D}}\hat{\mathcal{H}}\hat{\mathcal{D}}$, all coupling terms in the original Hamiltonian $\hat{\mathcal{H}}$ are removed. 
Indeed, by inserting $\mathcal{H}(\bullet)=[H,{}\bullet{}]$ with the decomposition of the Hamiltonian $H$ in~\eqref{Hamilton_decomp}, we get
\begin{align}
	(\mathcal{D}\mathcal{H}\mathcal{D})(\rho)
	&=\frac{1}{d_1}\mathbb{1}_1\otimes\tr_1\!\left[
	H_1\otimes\mathbb{1}_2+\mathbb{1}_1\otimes H_2+H_{12},\frac{1}{d_1}\mathbb{1}_1\otimes\tr_1\rho
	\right]\nonumber\\
	&=\frac{1}{d_1}\mathbb{1}_1\otimes[H_2,\tr_1\rho]\nonumber\\
	&=[(\mathbb{I}_1\otimes\mathcal{H}_2)\mathcal{D}](\rho)\label{Decoupling_by_av},
\end{align}
where $\mathcal{H}_2(\bullet)=[H_2,{}\bullet{}]$.
Therefore, 
\begin{equation}
\hat{\mathcal{E}}_n^\mathrm{av}(t)=\rme^{-\rmi t(\hat{\mathbb{I}}_1\otimes\hat{\mathcal{H}}_2)}\hat{\mathcal{D}}+\order{1/n},\label{eq:convergence_av}
\end{equation}
where the error is bounded by~\eqref{eq:av_error}. Notice that the Zeno dynamics $\mathcal{U}_Z(t)=\rme^{-\rmi t(\mathbb{I}_1\otimes\mathcal{H}_2)}\mathcal{D}$ gives rise to a unitary evolution of system 2.

We now prove Propositions~\ref{prop:Av2D} and~\ref{prop:Av2P}, which provide bounds on the evolution of system 2 in the Zeno limit.
\begin{proof}[Proof of Proposition~\ref{prop:Av2D}]
The target evolution $\mathcal{U}_Z(t)=\rme^{-\rmi t(\mathbb{I}_1\otimes\mathcal{H}_2)}\mathcal{D}$ is unitary on $\mathscr{H}_2$, and the Choi-Jamio\l{}kowski state of the reduced evolution of $\mathcal{U}_Z(t)$ for system 2 is pure, given by $\frac{1}{d_2}|\rme^{-\rmi tH_2})(\rme^{-\rmi tH_2}|$ according to~\eqref{eq:vectorization} and~\eqref{eq:Choi-state}. By definition, we have
\begin{align}
\left\|
\Lambda^\mathrm{av}_{2,\sigma_1}(n)
-
\frac{1}{d_2}|\rme^{-\rmi tH_2})(\rme^{-\rmi tH_2}|
\right\|_2
&=
\left\|
\tr_1\!\left[
\left(
[
\mathcal{E}^\mathrm{av}_n(t)
-
\rme^{-\rmi t(\mathbb{I}_1\otimes\mathcal{H}_2)}\mathcal{D}
]_{12}
\otimes
\mathbb{I}_{2'}
\right)\!
\left(
\sigma_1\otimes\frac{1}{d_2}|\mathbb{1}_2)_{22'}(\mathbb{1}_2|
\right)
\right]
\right\|_2
\nonumber\\
&\le
\sup_{\|A_{22'}\|_2=1}
\left\|
\tr_1\!\left[
\left(
[
\mathcal{E}^\mathrm{av}_n(t)
-
\rme^{-\rmi t(\mathbb{I}\otimes\mathcal{H}_2)}\mathcal{D}
]_{12}
\otimes
\mathbb{I}_{2'}
\right)\!
(
\sigma_1\otimes A_{22'}
)
\right]
\right\|_2,
\intertext{where in the second line the projection onto the maximally entangled state has been replaced by a supremum over all operators with norm 1. Now, by using Lemma~\ref{lemma:reduced_map} in Appendix~\ref{appendix:lemmas}, it is further bounded by}
&\le
\sqrt{d_1}\,
\|\sigma_1\|_2
\|
[
\hat{\mathcal{E}}^\mathrm{av}_n(t)
-
\rme^{-\rmi t(\hat{\mathbb{I}}_1\otimes\hat{\mathcal{H}}_2)}\hat{\mathcal{D}}
]_{12}
\otimes
\hat{\mathbb{I}}_{2'}
\|_\infty
\nonumber\\
&=
\sqrt{d_1}\,
\|\sigma_1\|_2
\|
\hat{\mathcal{E}}^\mathrm{av}_n(t)
-
\rme^{-\rmi t(\hat{\mathbb{I}}_1\otimes\hat{\mathcal{H}}_2)}\hat{\mathcal{D}}
\|_\infty
\nonumber\\
&=
\sqrt{d_1}\,
\|\sigma_1\|_2
\|
(\hat{\mathcal{D}}\rme^{-\rmi\frac{t}{n}\hat{\mathcal{H}}}\hat{\mathcal{D}})^n
-
\rme^{-\rmi t\hat{\mathcal{D}}\hat{\mathcal{H}}\hat{\mathcal{D}}}\hat{\mathcal{D}}
\|_\infty.
\end{align}
Then, the inequality~\eqref{eq:av_error} gives~\eqref{eq:DU2av} of Proposition~\ref{prop:Av2D}. The bound on the diamond distance in~\eqref{eq:DU2av_diamond} follows by using Lemma~\ref{lemma:diamond_choi} in Appendix~\ref{appendix:lemmas} together with the norm equivalence between the trace norm and the Frobenius norm in~\eqref{eq:equivalence_tr_fro}. 
\end{proof}

\begin{proof}[Proof of Proposition~\ref{prop:Av2P}]
By using Lemma~\ref{lemma:pur_ev} in Appendix~\ref{appendix:lemmas} and recalling the fact that $\infnorm{\Lambda}$ gives the largest singular value of $\Lambda$, we proceed as
\begin{align}
\sqrt{\mathtt{P}\bigl(\Lambda^\mathrm{av}_{2,\sigma_1}(n)\bigr)}
&\ge
\|\Lambda^\mathrm{av}_{2,\sigma_1}(n)\|_\infty
\nonumber\\
&\ge
\frac{1}{d_2}
(\rme^{-\rmi tH_2}|\Lambda^\mathrm{av}_{2,\sigma_1}(n)|\rme^{-\rmi tH_2}).
\nonumber\\
&=
1-
\frac{1}{d_2}
(\rme^{-\rmi tH_2}|
\left(
\frac{1}{d_2}|\rme^{-\rmi tH_2})(\rme^{-\rmi tH_2}|-\Lambda^\mathrm{av}_{2,\sigma_1}(n)
\right)
|\rme^{-\rmi tH_2})
\nonumber\\
&\ge
1-
\left\|
\frac{1}{d_2}|\rme^{-\rmi tH_2})(\rme^{-\rmi tH_2}|-\Lambda^\mathrm{av}_{2,\sigma_1}(n)
\right\|_\infty
\nonumber\\
&\ge
1-
\left\|
\frac{1}{d_2}|\rme^{-\rmi tH_2})(\rme^{-\rmi tH_2}|-\Lambda^\mathrm{av}_{2,\sigma_1}(n)
\right\|_2,
\nonumber\\
&\ge1-\frac{1}{n}\sqrt{d_1}\,\|\sigma_1\|_2T^2,
\end{align}
where we have used the norm equivalence~\eqref{eqn:NormEquivalence} for the second last inequality and Proposition~\ref{prop:Av2D} for the last inequality.
Proposition~\ref{prop:Av2P} is thus proved. 
\end{proof}

These results on the average dynamical decoupling evolution will help us to bound the efficiency of the random \emph{trajectory} dynamical decoupling evolution, by using the fact that the average evolution is a convex combination of different trajectory evolutions. Before turning to the trajectory case, let us briefly discuss the reduced Choi-Jami\l{}kowski states. They are the key players in the bounds presented in Sec.~\ref{sec:results}, and there are some subtleties to consider.

\section{Reduced Choi-Jamio\l{}kowski States}\label{sec:Choi}
For the sake of clarity, we introduce some notation first. Let $\mathcal{T}:\mathscr{H}\rightarrow\mathscr{H}$ be a CPTP map on a bipartite quantum system $\mathscr{H}=\mathscr{H}_1\otimes \mathscr{H}_2$. On one hand, the reduced maps $\mathcal{T}_{1,\sigma_2}:\mathscr{H}_1\rightarrow\mathscr{H}_1$ and $\mathcal{T}_{2,\sigma_1}:\mathscr{H}_2\rightarrow\mathscr{H}_2$ are given by
\begin{align}
\mathcal{T}_{1,\sigma_2}(\rho_1)&=\tr_2[\mathcal{T}(\rho_1\otimes\sigma_2)],\\
\mathcal{T}_{2,\sigma_1}(\rho_2)&=\tr_1[\mathcal{T}(\sigma_1\otimes\rho_2)],
\end{align}
and let $\Lambda_{1,\sigma_2}$ and $\Lambda_{2,\sigma_1}$ denote their corresponding Choi-Jamio\l{}kowski states. On the other hand, the Choi-Jamio\l{}kowski state $\Lambda$ of $\mathcal{T}$ lives on an enlarged Hilbert space $(\mathscr{H}_1\otimes\mathscr{H}_{1'})\otimes(\mathscr{H}_2\otimes\mathscr{H}_{2'})$, and we consider its reduced states
\begin{align}
\Lambda_{1}&=\tr_{22'}\Lambda,\\
\Lambda_{2}&=\tr_{11'}\Lambda.
\end{align}
These reduced Choi-Jamio\l{}kowski states $\Lambda_1$ and $\Lambda_2$ are related to the Choi-Jamio\l{}kowski states $\Lambda_{1,\sigma_2}$ and $\Lambda_{2,\sigma_1}$ of the reduced maps by the following lemma.
\begin{lemma}\label{LemmaReducedCJ}
	For the reduced Choi-Jamio\l{}kowski states of a CPTP map $\mathcal{T}:\mathscr{H}\rightarrow\mathscr{H}$ on a bipartite Hilbert space $\mathscr{H}=\mathscr{H}_1\otimes \mathscr{H}_2$, we have
\begin{align}
	\Lambda_1&=\Lambda_{1,\sigma_2}\quad\mathrm{with}\quad\sigma_2=\frac{1}{d_2}\mathbb{1}_2,\label{eq:reduced_choi_cond1}\\
	\Lambda_2&=\Lambda_{2,\sigma_1}\quad\mathrm{with}\quad\sigma_1=\frac{1}{d_1}\mathbb{1}_1.\label{eq:reduced_choi_cond2}
\end{align}
\end{lemma}
\begin{proof}
		Let $\{\ket{i}\}_{i=1,\ldots,d_1}$ and $\{\ket{i'}\}_{i'=1,\ldots,d_1}$ be orthonormal bases of $\mathscr{H}_1$ and $\mathscr{H}_{1'}$, respectively. Analogously, $\{\ket{j}\}_{j=1,\ldots,d_2}$ and $\{\ket{j'}\}_{j'=1,\ldots,d_2}$ for $\mathscr{H}_2$ and $\mathscr{H}_{2'}$, respectively.
		Let $\mathcal{T}(\rho)=\sum_kE_k\rho E_k^\dag$ be a Kraus representation of the CPTP map $\mathcal{T}$.
		Then, on one hand, the reduced CPTP map $\mathcal{T}_{1,\sigma_2}$, which only acts on $\mathscr{H}_1$, is given by
		\begin{equation}
			\mathcal{T}_{1,\sigma_2}(\rho_1)=\sum_k \tr_2[E_k(\rho_1\otimes\sigma_2)E_k^\dag],
		\end{equation}
		and its Choi-Jamio\l{}kowski state reads
		\begin{equation}
			\Lambda_{1,\sigma_2}=\frac{1}{d_1}\sum_{i,i'}\sum_k \tr_2[E_k(\ket{i}\bra{i'}\otimes\sigma_2)E_k^\dag]\otimes\ket{i}\bra{i'}.\label{eq:CPTP_reduced_Choi}
		\end{equation}
		On the other hand, the reduced Choi-Jamio\l{}kowski state $\Lambda_1$ takes the form
		\begin{align}
				\Lambda_{1}&=\frac{1}{d_1d_2}\sum_{i,i',j,j'}\sum_{k} \tr_{22'}[E_k(\ket{i}\bra{i'}\otimes\ket{j}\bra{j'})E_k^\dag\otimes(\ket{i}\bra{i'}\otimes\ket{j}\bra{j'})]\nonumber\\
				&=\frac{1}{d_1d_2}\sum_{i,i',j}\sum_k \tr_2[E_k(\ket{i}\bra{i'}\otimes\ket{j}\bra{j})E_k^\dagger]\otimes\ket{i}\bra{i'}\nonumber\\
				&=\frac{1}{d_1}\sum_{i,i'}\sum_k \tr_2\!\left[
				E_k\left(\ket{i}\bra{i'}\otimes\frac{1}{d_2}\mathbb{1}_2\right)E_k^\dagger
				\right]
				\otimes\ket{i}\bra{i'}.
				\label{eq:reduced_CPTP}
		\end{align}
		Comparing~\eqref{eq:CPTP_reduced_Choi} and~\eqref{eq:reduced_CPTP}, we see that the relation~\eqref{eq:reduced_choi_cond1} holds.
		The relation~\eqref{eq:reduced_choi_cond2} is also confirmed completely in the same way.
\end{proof}
\noindent
Lemma~\ref{LemmaReducedCJ} shows that the reduced Choi-Jamio\l{}kowsi state $\Lambda_{1(2)}$ is the Choi-Jamio\l{}kowski state $\Lambda_{1(2),\sigma_{2(1)}}$ of the reduced map  of system 1(2) when system 2(1) starts in the maximally mixed state $\sigma_{2(1)}=\mathbb{1}_{2(1)}/d_{2(1)}$.
In order to bridge between $\Lambda_{1(2)}$ and $\Lambda_{1(2),\sigma_{2(1)}}$ for a general initial state $\sigma_{2(1)}$, we will use the following two lemmas.
\begin{lemma}\label{lemma:max_mixed}
Let $\mathscr{H}$ be a Hilbert space of dimension $d$ and let $\sigma$ be a density operator on $\mathscr{H}$ with spectral decomposition $\sigma=\sum_i s_i\ket{i}\bra{i}$. Then, the maximally mixed state $\mathbb{1}/d$ on $\mathscr{H}$ can be decomposed as
\begin{equation}
\frac{1}{d}\mathbb{1}
=\frac{1}{d\|\sigma\|_\infty}\sigma+\left(1-\frac{1}{d\|\sigma\|_\infty}\right)\omega,
\end{equation}
with another density operator
\begin{equation}
\omega=\frac{1}{d-1/\|\sigma\|_\infty}
\sum_i\left(1-\frac{s_i}{\|\sigma\|_\infty}\right)\ket{i}\bra{i}.
\label{eqn:omega}
\end{equation}
\end{lemma}
\begin{proof}
The orthonormalized vectors $\{\ket{i}\}_{i=1,\ldots,d}$ in the spectral decomposition of $\sigma$
form a basis of $\mathscr{H}$. 
The maximally mixed state $\mathbb{1}/d$ can be expanded in this basis as
\begin{equation}
\frac{1}{d}\mathbb{1}
=\frac{1}{d}\sum_i\ket{i}\bra{i}.
\end{equation}
In order to find a convex decomposition of the maximally mixed state $\mathbb{1}/d$ into the state $\sigma$ and another state $\omega$, let us rewrite the expansion of the maximally mixed state $\mathbb{1}/d$ as
\begin{equation}
\frac{1}{d}\mathbb{1}
=p\sigma+\sum_i\left(\frac{1}{d}-ps_i\right)\ket{i}\bra{i}.
\end{equation}
In order for this to be a convex decomposition, we should have $1/d-ps_i\ge0$, $\forall i$.
That is, we need
\begin{equation}
p\le\frac{1}{ds_i},\quad\forall i.
\end{equation}
Therefore, let us choose
\begin{equation}
p=\frac{1}{d\|\sigma\|_\infty}.
\end{equation}
Then,
\begin{align}
\frac{1}{d}\mathbb{1}
&=\frac{1}{d\|\sigma\|_\infty}\sigma+\frac{1}{d}\sum_i\left(1-\frac{s_i}{\|\sigma\|_\infty}\right)\ket{i}\bra{i}
\nonumber\\
&=\frac{1}{d\|\sigma\|_\infty}\sigma+\left(1-\frac{1}{d\|\sigma\|_\infty}\right)\omega,
\end{align}
with $\omega$ defined in~\eqref{eqn:omega}.
\end{proof}
\begin{lemma}\label{lemma:Choi}
Let $\mathscr{H}_1$ and $\mathscr{H}_{1'}$ be Hilbert spaces of dimension $d_1$, and let $\mathscr{H}_2$ and $\mathscr{H}_{2'}$ be Hilbert spaces of dimension $d_2$. Let $|v)$ denote a normalized vector in the doubled Hilbert space $\mathscr{H}_1\otimes\mathscr{H}_{1'}$, and analogously $|w)$ a normalized vector in $\mathscr{H}_2\otimes\mathscr{H}_{2'}$. Let $\mathcal{T}$ be a CPTP map acting on operators on the bipartite Hilbert space $\mathscr{H}=\mathscr{H}_1\otimes \mathscr{H}_2$. Then, for its reduced Choi-Jamio\l{}kowski states, we have
\begin{align}
(v|\Lambda_{1,\sigma_2}|v)
&\ge
1-d_2\|\sigma_2\|_\infty[1-(v|\Lambda_1|v)],\\
(w|\Lambda_{2,\sigma_1}|w)
&\ge
1-d_1\|\sigma_1\|_\infty[1-(w|\Lambda_2|w)],
\end{align}
and
\begin{align}
\|\Lambda_{1,\sigma_2}\|_\infty
&\ge
1-d_2\|\sigma_2\|_\infty(1-\|\Lambda_1\|_\infty),\\
\|\Lambda_{2,\sigma_1}\|_\infty
&\ge
1-d_1\|\sigma_1\|_\infty(1-\|\Lambda_2\|_\infty).
\end{align}
\end{lemma}
\begin{proof}
We focus on the proof for $\Lambda_{1,\sigma_2}$ on system 1. 
The relations for $\Lambda_{2,\sigma_1}$ on system 2 can be proved in the same way.
By using the convex decomposition of the maximally mixed state from Lemma~\ref{lemma:max_mixed},
\begin{equation}
\frac{1}{d_2}\mathbb{1}_2
=\frac{1}{d_2\|\sigma_2\|_\infty}\sigma_2+\left(1-\frac{1}{d_2\|\sigma_2\|_\infty}\right)\omega_2,
\end{equation}
we have, by linearity and Lemma~\ref{LemmaReducedCJ},
\begin{align}
(v|\Lambda_1|v)
&=\frac{1}{d_2\|\sigma_2\|_\infty}(v|\Lambda_{1,\sigma_2}|v)
+\left(1-\frac{1}{d_2\|\sigma_2\|_\infty}\right)(v|\Lambda_{1,\omega_2}|v)
\nonumber\\
&\le\frac{1}{d_2\|\sigma_2\|_\infty}(v|\Lambda_{1,\sigma_2}|v)
+1-\frac{1}{d_2\|\sigma_2\|_\infty},
\end{align}
and hence,
\begin{equation}
(v|\Lambda_{1,\sigma_2}|v)
\ge
1-d_2\|\sigma_2\|_\infty[1-(v|\Lambda_1|v)],
\end{equation}
for any normalized vector $|v)$.
By choosing $|v)$ to be the normalized eigenvector $|v_\mathrm{max})$ of $\Lambda_1$ belonging to its largest eigenvalue, we get
\begin{align}
\|\Lambda_{1,\sigma_2}\|_\infty
&\ge(v_\mathrm{max}|\Lambda_{1,\sigma_2}|v_\mathrm{max})
\nonumber\\
&\ge
1-d_2\|\sigma_2\|_\infty[1-(v_\mathrm{max}|\Lambda_1|v_\mathrm{max})]
\nonumber\\
&=
1-d_2\|\sigma_2\|_\infty(1-\|\Lambda_1\|_\infty).
\end{align}
\end{proof}

\begin{remark}\label{remark:reduced_pur}
	Lemma~\ref{lemma:Choi} can alternatively be stated in terms of the purities of the Choi-Jamio\l{}kowski states as
	\begin{align}
		\sqrt{\pur(\Lambda_{1,\sigma_2})}&\geq 1-d_2\infnorm{\sigma_2}[1-\pur(\Lambda_1)],\label{eq:reduced_pur_1}\\
		\sqrt{\pur(\Lambda_{2,\sigma_1})}&\geq 1-d_1\infnorm{\sigma_1}[1-\pur(\Lambda_2)].
	\end{align}
	This can easily be verified by using Lemma~\ref{lemma:pur_ev} in Appendix~\ref{appendix:lemmas}, in particular, $\sqrt{\pur(\Lambda_{1,\sigma_2})}\geq\|\Lambda_{1,\sigma_2}\|_\infty$ and $\|\Lambda_1\|_\infty\geq\pur(\Lambda_1)$ (and analogously for $\Lambda_{2,\sigma_1}$).
\end{remark}
With this knowledge about the reduced Choi-Jamio\l{}kowsi states at hand, we are now ready to prove the statements about the random trajectory dynamical decoupling given in Sec.~\ref{sec:results}.

\section{Convergence of the Trajectory Dynamical Decoupling}\label{sec:conergence_tr}
This section provides the proofs of the theorems presented in Secs.~\ref{sec:Tr2} and~\ref{sec:Tr1}. We start by proving the statements about the convergence of the random trajectory evolution on system 2 under the dynamical decoupling pulses, i.e.~Theorems~\ref{thm:Tr2P} and~\ref{thm:Tr2D} presented in Sec.~\ref{sec:Tr2}. We also provide a proof of Proposition~\ref{prop:Tr2Tail}, which gives a bound on the probability of trajectory evolutions on system 2 being far from the unitary $\rme^{-\rmi tH_2}$. Afterwards, we turn our attention to the proofs of the statements about the convergence of the random trajectory evolution on system 1, presented in Sec.~\ref{sec:Tr1}. That is, we prove Theorems~\ref{thm:Tr1P} and~\ref{thm:Tr1D}. 
Finally, we prove Proposition~\ref{thm:dist_unitary_purity}. The corollaries presented in Sec.~\ref{sec:results} are then direct consequences of the theorems and the proposition proved here. Corollary~\ref{cor:Tr2V} follows from Theorems~\ref{thm:Tr2P} and~\ref{thm:Tr2D} using the Bhatia-Davis inequality $\Var[X]\leq(\max[X]-\E[X])(\E[X]-\min[X])$~\cite{BD2000}. In the same way, Corollary~\ref{cor:Tr1V} follows from Theorems~\ref{thm:Tr1P} and~\ref{thm:Tr1D}. In addition, combining Theorem~\ref{thm:Tr1P} with Proposition~\ref{thm:dist_unitary_purity} leads to Corollary~\ref{cor:Tr1D}.

\subsection{Convergence of the Trajectory Evolution of System 2}\label{sec:proofTr2}
Theorem~\ref{thm:Tr2P} gives a bound on the expected purity, the expected operator norm, and the expected fidelity to the unitary $\rme^{-\rmi tH_2}$ of the reduced Choi-Jamio\l{}kowski state of the trajectory evolution on system 2.
\begin{proof}[Proof of Theorem~\ref{thm:Tr2P}]
First, notice that according to Lemma~\ref{lemma:pur_ev} in Appendix~\ref{appendix:lemmas} we have
\begin{equation}
\sqrt{\mathtt{P}\bigl(\Lambda_{2,\sigma_1}^{(j)}(n)\bigr)}
\ge
\|\Lambda_{2,\sigma_1}^{(j)}(n)\|_\infty
\ge\frac{1}{d_2}(\rme^{-\rmi tH_2}|\Lambda_{2,\sigma_1}^{(j)}(n)|\rme^{-\rmi tH_2}).\label{eq:pur_opnorm_fid}
\end{equation}
By Jensen's inequality $\E[f({}\cdots{})]\le f(\E[{}\cdots{}])$ for concave functions $f(x)$, we can bound the expected purity, the expected operator norm, and the expected fidelity of the reduced Choi-Jamio\l{}kowski state of the trajectory evolution on system 2 as
\begin{align}
\sqrt{
\mathbb{E}\!\left[
\mathtt{P}\bigl(\Lambda_{2,\sigma_1}^{(j)}(n)\bigr)
\right]
}
&\ge
\mathbb{E}\!\left[
\sqrt{\mathtt{P}\bigl(\Lambda_{2,\sigma_1}^{(j)}(n)\bigr)}
\right]
\nonumber\\
&\ge
\mathbb{E}\Bigl[\|\Lambda_{2,\sigma_1}^{(j)}(n)\|_\infty\Bigr]
\nonumber\\
&\ge\mathbb{E}\!\left[
\frac{1}{d_2}(\rme^{-\rmi tH_2}|\Lambda_{2,\sigma_1}^{(j)}(n)|\rme^{-\rmi tH_2})
\right]
\nonumber\\
&=
\frac{1}{d_2}(\rme^{-\rmi tH_2}|\Lambda^\mathrm{av}_{2,\sigma_1}(n)|\rme^{-\rmi tH_2})
\nonumber\\
&\ge1-\frac{1}{n}\sqrt{d_1}\,\|\sigma_1\|_2T^2,
\end{align}
where we have used the linearity of the expectation value for the equality and the last inequality follows from Proposition~\ref{prop:Av2P}.
\end{proof}
\noindent
Theorem~\ref{thm:Tr2P} can be used to bound the expected distance of the Choi-Jamio\l{}kowski state of system 2 to the Zeno unitary, stated in Theorem~\ref{thm:Tr2D}.
\begin{proof}[Proof of Theorem~\ref{thm:Tr2D}]
First, observe that for each trajectory $j$ we have by definition
\begin{equation}
\left\|
\Lambda_{2,\sigma_1}^{(j)}(n)
-\frac{1}{d_2}|\rme^{-\rmi tH_2})(\rme^{-\rmi tH_2}|
\right\|_2^2
=
\mathtt{P}\bigl(\Lambda_{2,\sigma_1}^{(j)}(n)\bigr)
-\frac{2}{d_2}
(\rme^{-\rmi tH_2}|
\Lambda_{2,\sigma_1}^{(j)}(n)
|\rme^{-\rmi tH_2})
+1.
\label{eq:distU-pur}
\end{equation}
Now, consider the expected distance to the Zeno unitary. By Jensen's inequality for convex functions, we get
\begin{align}
\left(
\mathbb{E}\!\left[
\left\|
\Lambda_{2,\sigma_1}^{(j)}(n)
-\frac{1}{d_2}|\rme^{-\rmi tH_2})(\rme^{-\rmi tH_2}|
\right\|_2
\right]
\right)^2
&\le
\mathbb{E}\!\left[
\left\|
\Lambda_{2,\sigma_1}^{(j)}(n)
-\frac{1}{d_2}|\rme^{-\rmi tH_2})(\rme^{-\rmi tH_2}|
\right\|_2^2
\right]
\nonumber\\
&=
\mathbb{E}\!\left[
\mathtt{P}\bigl(\Lambda_{2,\sigma_1}^{(j)}(n)\bigr)
\right]
-\frac{2}{d_2}
(\rme^{-\rmi tH_2}|
\Lambda^\mathrm{av}_{2,\sigma_1}(n)
|\rme^{-\rmi tH_2})
+1
\nonumber\\
&\le
2\left(
1-\frac{1}{d_2}
(\rme^{-\rmi tH_2}|
\Lambda^\mathrm{av}_{2,\sigma_1}(n)
|\rme^{-\rmi tH_2})
\right)
\nonumber\\
&\le
\frac{2}{n}\sqrt{d_1}\,\|\sigma_1\|_2T^2,
\end{align}
where we have bounded the purity of the Choi-Jamio\l{}kowski state by 1 and used Proposition~\ref{prop:Av2P}.
This proves the bound~\eqref{eq:DU2tr} on the distance in terms of the Choi-Jamio\l{}kowski state.
The other bound~\eqref{eq:DU2tr_diamond} on the diamond distance is then obtained by using Lemma~\ref{lemma:diamond_choi} in Appendix~\ref{appendix:lemmas} to translate the diamond distance into the trace distance in terms of the Choi-Jamio\l{}kowski state and by using the norm equivalence~\eqref{eq:equivalence_tr_fro} to bound the trace distance by the Frobenius distance available in~\eqref{eq:DU2tr}.
\end{proof}

These proofs are based on the fact that the average dynamical decoupling evolution on system 2 is an arithmetic mean of the trajectory dynamical decoupling evolutions on system 2, and hence, the Choi-Jamio\l{}kowski state of the average evolution is given by a \emph{convex sum} of the Choi-Jamio\l{}kowski states of the trajectory evolutions, i.e.,
\begin{equation}
\Lambda^\mathrm{av}_{2,\sigma_1}(n)
=\sum_j
q_j\Lambda_{2,\sigma_1}^{(j)}(n),\label{eq:convex_choi_2}
\end{equation}
with $q_j$ being the probability of realizing trajectory $j$ (which is actually uniformly distributed).
As proved in Propositions~\ref{prop:Av2D} and~\ref{prop:Av2P}, the average dynamical decoupling evolution on system 2 converges to the Zeno unitary in the decoupling limit $n\rightarrow\infty$. This leads to the results proved in Theorems~\ref{thm:Tr2P} and~\ref{thm:Tr2D}, which show that the trajectory dynamical decoupling evolutions on system 2 converge to the Zeno unitary \emph{on average}. This further ensures that \emph{almost all} trajectory dynamical decoupling evolutions on system 2 converge to the same Zeno unitary, due to the extremality of the pure Choi-Jamio\l{}kowski state of a unitary. This is made explicit in a quantitative way by Proposition~\ref{prop:Tr2Tail}.
\begin{proof}[Proof of Proposition~\ref{prop:Tr2Tail}]
Consider the fidelity of the reduced Choi-Jamiol\l{}kowski state of the average dynamical decoupling evolution on system 2 to the unitary $\rme^{-\rmi tH_2}$. It is expressed, using the convex sum~\eqref{eq:convex_choi_2}, as
\begin{align}
\frac{1}{d_2}(\rme^{-\rmi tH_2}|\Lambda^\mathrm{av}_{2,\sigma_1}(n)|\rme^{-\rmi tH_2})
&=
\sum_j
q_j
\frac{1}{d_2}
(\rme^{-\rmi tH_2}|\Lambda_{2,\sigma_1}^{(j)}(n)|\rme^{-\rmi tH_2}).
\intertext{Let $r_j=\frac{1}{d_2}(\rme^{-\rmi tH_2}|\Lambda_{2,\sigma_1}^{(j)}(n)|\rme^{-\rmi tH_2})$ be the fidelity of trajectory $j$ and fix $r\in[0,1]$. Then,}
\frac{1}{d_2}(\rme^{-\rmi tH_2}|\Lambda^\mathrm{av}_{2,\sigma_1}(n)|\rme^{-\rmi tH_2})
&=
\sum_j
q_j
r_j
\nonumber\\
&=
\sum_{r_j>r}
q_j
r_j
+
\sum_{r_j\le r}
q_j
r_j
\nonumber\\
&\le
\sum_{r_j>r}
q_j
+
\sum_{r_j\le r}
q_j
r
\nonumber\\
&=1-\mathbb{P}(\le r)+\mathbb{P}(\le r)r,
\end{align}
where $\mathbb{P}(\le r)=\sum_{r_j\le r} q_j$ is the probability of finding $r_j\le r$. Notice that we have used $\sum_{r_j>r}q_j=1-\mathbb{P}(\le r)$. Rearranging this inequality and using Proposition~\ref{prop:Av2P}, we get
\begin{align}
\mathbb{P}(\le r)
&\le\frac{1}{1-r}\left(
1-\frac{1}{d_2}(\rme^{-\rmi tH_2}|\Lambda^\mathrm{av}_{2,\sigma_1}(n)|\rme^{-\rmi tH_2})
\right)
\nonumber\\
&\le
\frac{1}{1-r}\frac{1}{n}\sqrt{d_1}\,\|\sigma_1\|_2T^2,
\end{align}
and thereby the statement of the proposition is proved.
\end{proof}

Now that we have proved the convergence bounds for the trajectory dynamical decoupling evolutions on system 2, we move on to proving the convergence of the trajectory dynamical decoupling evolutions on system 1.

\subsection{Convergence of the Trajectory Evolution of System 1}\label{sec:proofTr1}
Recall that for each trajectory $j$ the evolution of the total system $\mathcal{E}_n^{(j)}(t)$ is unitary. Therefore, its Choi-Jamio\l{}kowski state $\Lambda^{(j)}(n)$ is pure. As a consequence, the reduced Choi-Jamio\l{}kowski states $\Lambda_1^{(j)}(n)$ and $\Lambda_2^{(j)}(n)$ share the same spectrum given by the Schmidt coefficients of the full state $\Lambda^{(j)}(n)$. This makes it possible to shift the error bounds on system 2 to system 1. By doing so, we prove Theorem~\ref{thm:Tr1P}, which gives a bound on the expected purity and the expected operator norm of the reduced Choi-Jamio\l{}kowski state of system 1, as well as Theorem~\ref{thm:Tr1D}, which provides a bound on the expected distance of the reduced Choi-Jamio\l{}kowski state of the trajectory evolution of system 1 to a pure state.
\begin{proof}[Proof of Theorem~\ref{thm:Tr1P}]
First, notice that for each trajectory $j$ we can apply Lemma~\ref{lemma:Choi}.
In addition, we use $\|\Lambda_1^{(j)}(n)\|_\infty=\|\Lambda_2^{(j)}(n)\|_\infty$ [see~\eqref{eq:Schmidt_OpNorm}], and get
\begin{align}
\|\Lambda_{1,\sigma_2}^{(j)}(n)\|_\infty
&\ge1-d_2\|\sigma_2\|_\infty(1-\|\Lambda_1^{(j)}(n)\|_\infty)
\nonumber\\
&=1-d_2\|\sigma_2\|_\infty(1-\|\Lambda_2^{(j)}(n)\|_\infty).
\label{eq:OpNorm12}
\end{align}
Now, consider the square root of the expected purity of the Choi-Jamio\l{}kowski state $\Lambda_{1,\sigma_2}^{(j)}(n)$ of the trajectory evolution of system 1, and use Jensen's inequality for the square root (which is concave) to obtain
\begin{align}
\sqrt{
\mathbb{E}\!\left[
\mathtt{P}\bigl(\Lambda_{1,\sigma_2}^{(j)}(n)\bigr)
\right]
}
&\ge
\mathbb{E}\!\left[
\sqrt{\mathtt{P}\bigl(\Lambda_{1,\sigma_2}^{(j)}(n)\bigr)}
\right]
\nonumber\\
&\ge
\mathbb{E}\Bigl[\|\Lambda_{1,\sigma_2}^{(j)}(n)\|_\infty\Bigr]
\nonumber\\
&\ge1-d_2\|\sigma_2\|_\infty\left(
1-\mathbb{E}\Bigl[\|\Lambda_2^{(j)}(n)\|_\infty\Bigr]
\right)
\nonumber\\
&\ge1-\frac{1}{n}d_2\|\sigma_2\|_\infty T^2.
\end{align}
The second inequality follows from Lemma~\ref{lemma:pur_ev} in Appendix~\ref{appendix:lemmas}, the third inequality is due to~\eqref{eq:OpNorm12} with the linearity of the expectation value, and the last bound is a consequence of Theorem~\ref{thm:Tr2P} for $\sigma_1=\mathbb{1}/d_1$ (noting Lemma~\ref{LemmaReducedCJ}), for which we have $\|\sigma_1\|_2=1/\sqrt{d_1}$.
\end{proof}

\begin{proof}[Proof of Theorem~\ref{thm:Tr1D}]
Without loss of generality, let us label the eigenvalues $\{\lambda_k\}$ of $\Lambda_{1,\sigma_2}^{(j)}(n)$ in a descending order as $\lambda_0\ge\lambda_1\ge\cdots\ge\lambda_{d_1^2-1}$. Then, $|v_{1,\sigma_2}^{(j)}(n))$ is the normalized eigenvector of $\Lambda_{1,\sigma_2}^{(j)}(n)$ belonging to the largest eigenvalue $\lambda_0$, and we have
\begin{equation}
\mathbb{E}\!\left[
\left\|
\Lambda_{1,\sigma_2}^{(j)}(n)
-|v_{1,\sigma_2}^{(j)}(n))(v_{1,\sigma_2}^{(j)}(n)|
\right\|_\infty
\right]
=
\mathbb{E}[
\max(1-\lambda_0,\lambda_1)
].
\end{equation}
However, since $\lambda_1\le \sum_{k\ge1}\lambda_k =1-\lambda_0= 1-\|\Lambda_{1,\sigma_2}^{(j)}(n)\|_\infty$, we get
\begin{align}
\mathbb{E}\!\left[
\left\|
\Lambda_{1,\sigma_2}^{(j)}(n)
-|v_{1,\sigma_2}^{(j)}(n))(v_{1,\sigma_2}^{(j)}(n)|
\right\|_\infty
\right]
&= 
1-\mathbb{E}\Bigl[
\|\Lambda_{1,\sigma_2}^{(j)}(n)\|_\infty
\Bigr]\nonumber\\
&\le
\frac{1}{n}d_2\|\sigma_2\|_\infty T^2,
\label{eqn:D1ChoiProof}
\end{align}
where the last inequality is due to Theorem~\ref{thm:Tr1P}. 
\end{proof}

So far, we have presented the convergence bounds in terms of the Choi-Jamio\l{}kowsi states of the reduced evolutions. We also give a bound in terms of the matrix representation of the reduced CPTP maps and a bound in terms of the diamond distance in Corollary~\ref{cor:Tr1D}.
It is done with the help of Proposition~\ref{thm:dist_unitary_purity}, which provides lower and upper bounds on the distance of a CPTP map $\mathcal{T}$ to a unitary $\mathcal{U}$ in terms of the matrix representation of the maps and an upper bound on the diamond distance between $\mathcal{T}$ and $\mathcal{U}$.
\begin{proof}[Proof of Proposition~\ref{thm:dist_unitary_purity}]
The idea of the proof is to use a relation between the Kraus operators $E_k$ of a CPTP map $\mathcal{T}(\rho)=\sum_kE_k\rho E_k^\dag$ and the spectral decomposition of the Choi-Jamio\l{}kowski state $\Lambda=\sum_k \lambda_k |v_k)(v_k|$. If $\mathcal{T}$ is unitary, then its Kraus representation reduces to $\mathcal{T}=E \rho E^\dagger$ with $E^\dagger E=\mathbb{1}$. In this unitary case, the largest eigenvalue of $\Lambda$ is $\lambda_\mathrm{max}=1$. For a nonunitary case, we have $\lambda_\mathrm{max}<1$.
Let us fix a Kraus representation as indicated in~\eqref{Kraus-choice}.
By this, we have $\tr(E_k^\dagger E_\ell)=d\lambda_k \delta_{k\ell}$. Without loss of generality, let $\lambda_0=\lambda_\mathrm{max}$ be the largest eigenvalue of $\Lambda$, and $\lambda_k\le\lambda_0$ for $k=1,\ldots,d^2-1$.
\begin{itemize}
\item Upper bound in~\eqref{eq:dist_unitary_purity}: 
Using the triangle inequality,
		\begin{equation}
		\Fnorm{\hat{\mathcal{T}}-\hat{\mathcal{U}}}\leq \Fnorm{\hat{\mathcal{T}}-E_0\otimes\overline{E_0}}+\Fnorm{E_0\otimes\overline{E_0}-\hat{\mathcal{U}}}.
		\end{equation}
The first term can be estimated in the following way.
		\begin{align}
			\Fnorm{\hat{\mathcal{T}}-E_0\otimes\overline{E_0}}^2
			&=\left\|
			\sum_{k\ge1} E_k\otimes \overline{E_k}
			\right\|_2^2\\
			&=d^2\sum_{k\ge1}\lambda_k^2\\
			&=d^2(\pur -\lambda_0^2)\vphantom{\sum_{k\ge1}}\\
			&\leq d^2(\pur-\pur^2),
			\label{eqn:Distance1st}
		\end{align}
		where we have followed the same steps as~\eqref{eq:Fnorm_purity} and have used Lemma~\ref{lemma:pur_ev} in Appendix~\ref{appendix:lemmas}\@. 
		For the second term, we make use of the polar decomposition of $E_0=U|E_0|$~\cite[Sec.~7.3]{Horn2012}, where $U$ is the closest unitary from $E_0$~\cite[Theorem~IX.7.2]{Bhatia1997}. 
		Exploiting the fact that $\|{}\bullet{}\|_2$ is unitarily invariant, and recalling the Kraus condition $\sum_kE_k^\dag E_k=\mathbb{1}$, which implies $|E_0|^2\le|E_0|\le\mathbb{1}$, we can bound the second term as
		\begin{align}
				\Fnorm{E_0\otimes\overline{E_0}-\hat{\mathcal{U}}}^2
				&=\Fnorm{E_0\otimes\overline{E_0}-U\otimes\overline{U}}^2
				\nonumber\\
				&=\Fnorm{\mathbb{1}\otimes\mathbb{1}-|E_0|\otimes\overline{|E_0|}}^2
				\nonumber\\
				&=\tr\Bigl(
				\mathbb{1}\otimes\mathbb{1}
				-2|E_0|\otimes\overline{|E_0|}
				+|E_0|^2\otimes\overline{|E_0|^2}
				\Bigr)\nonumber
				\\
				&\le\tr\Bigl(\mathbb{1}\otimes\mathbb{1}-|E_0|^2\otimes\overline{|E_0|^2}\Bigr)\nonumber
				\\
				&=d^2(1-\lambda_0^2)\nonumber
				\\
				&\leq d^2(1-\pur^2),
				\label{eqn:Distance2ndPrime}
		\end{align}
		where we have used $\tr(|E_0|^2)=\tr(E_0^\dag E_0)=d\lambda_0$ for the second last line, and Lemma~\ref{lemma:pur_ev} in Appendix~\ref{appendix:lemmas} for the last inequality. 
		Combining~\eqref{eqn:Distance1st} and~\eqref{eqn:Distance2ndPrime}, the upper bound of the proposition is proved.

\item Lower bound in~\eqref{eq:dist_unitary_purity}:
This follows by direct calculation.
		\begin{align}
				\Fnorm{\hat{\mathcal{T}}-\hat{\mathcal{U}}}
				&\geq \Fnorm{U\otimes \overline{U}}-\left\|\sum_{k} E_k\otimes \overline{E_k}\right\|_2\\
				&=d-\sqrt{d^2\sum_k\lambda_k^2}\\
				&=d(1-\sqrt{\pur}).
		\end{align}

\item Bound~\eqref{eq:diamond_dist_choi}:
By the triangle inequality, we can write
		\begin{equation}
			\|\mathcal{T}-\mathcal{U}\|_\diamond
			\le
			\|\mathcal{T}-E_0\bullet E_0^\dagger\|_\diamond
			+
			\|E_0\bullet E_0^\dagger - U\bullet U^\dagger\|_\diamond.\label{eq:triangle_diamond_norm}
		\end{equation}
		Let us bound both terms in~\eqref{eq:triangle_diamond_norm} individually, starting with the first one.
		By using Lemma~\ref{lemma:diamond_choi} and~\eqref{eq:vectorization},
		\begin{align}
			\|\mathcal{T}-E_0\bullet E_0^\dagger\|_\diamond
			&\le
			d\left\|\Lambda-\frac{1}{d}|E_0)(E_0|\right\|_1
			\nonumber\displaybreak[0]\\
			&=
			d\,\Bigl\|
			\Lambda-\lambda_0|v_0)(v_0|
			\Bigr\|_1
			\nonumber\displaybreak[0]\\
			&=
			d\left\|
			\sum_{k\ge1}\lambda_k|v_k)(v_k|
			\right\|_1
			\nonumber\displaybreak[0]\\
			&=
			d\sum_{k\ge1}\lambda_k
			\nonumber\displaybreak[0]\\
			&=
			d(1-\lambda_0)
			\nonumber\displaybreak[0]\\
			&=
			d(1-\|\Lambda\|_\infty),
		\end{align}
		where from the third to the fourth line we used the fact that for a positive matrix $A$ we have $\|A\|_1=\tr A$.
		For the second term in~\eqref{eq:triangle_diamond_norm}, we first recall the definition of the diamond norm
		\begin{equation}
			\|E_0\bullet E_0^\dagger - U\bullet U^\dagger\|_\diamond
			=
			\sup_{\|A\|_1=1}\|(E_0\otimes\mathbb{1})A(E_0^\dagger\otimes\mathbb{1})-(U\otimes\mathbb{1})A(U^\dagger\otimes\mathbb{1})\|_1,
		\end{equation}
		where $\mathbb{1}$ is the $d\times d$ identity matrix. Then, using the triangle inequality and the unitary invariance of the trace norm,
		\begin{equation}
			\|E_0\bullet E_0^\dagger - U\bullet U^\dagger\|_\diamond
			\le
			\sup_{\|A\|_1=1}\left(
			\Bigl\|[(E_0-U)\otimes\mathbb{1}]A(E_0^\dagger\otimes \mathbb{1})\Bigr\|_1
			+
			\Bigl\|A[(E_0^\dagger-U^\dagger)\otimes\mathbb{1}]\Bigr\|_1
			\right).
		\end{equation}
		By using H\"older's inequality twice,
		\begin{align}
			\|E_0\bullet E_0^\dagger - U\bullet U^\dagger\|_\diamond
			&\le
			\sup_{\|A\|_1=1}\left(
			\|(E_0-U)\otimes\mathbb{1}\|_\infty\|A\|_1\|E_0^\dagger\otimes\mathbb{1}\|_\infty
			+
			\|A\|_1\|(E_0^\dagger-U^\dagger)\otimes\mathbb{1}\|_\infty
			\right)\nonumber\\
			&=
			\|E_0-U\|_\infty (1+\|E_0\|_\infty)\nonumber\\
			&=
			\|\mathbb{1}-|E_0|\|_\infty (1+\|E_0\|_\infty),
		\end{align}
		where in the last step we wrote $E_0$ in its polar decomposition $E_0=U|E_0|$ \cite[Sec.~7.3]{Horn2012} with $U$ the closest unitary to $E_0$ \cite[Theorem.~IX.7.2]{Bhatia1997} and $|E_0|=\sqrt{E_0^\dagger E_0}$. Also, we used the unitary invariance of the Frobenius norm. Rewriting this expression while recalling the Kraus condition $\sum_kE_k^\dagger E_k=\mathbb{1}$, which implies $|E_0|^2\le|E_0|\le \mathbb{1}$, yields
		\begin{align}
			\|E_0\bullet E_0^\dagger - U\bullet U^\dagger\|_\diamond
			&\le
			\|\mathbb{1}-|E_0|^2\|_\infty (1+\|\mathbb{1}\|_\infty)\nonumber\\
			&=
			2\|\mathbb{1}-E_0^\dagger E_0\|_\infty\nonumber\\
			&\le
			2\|\mathbb{1}-E_0^\dagger E_0\|_1\nonumber\\
			&=
			2\left\|\sum_{k\ge 1} E_k^\dagger E_k\right\|_1\nonumber\\
			&=
			2\tr\sum_{k\ge1}E_k^\dag E_k\nonumber\\
			&=
			2\sum_{k\ge1}(E_k|E_k)\nonumber\\
			&=
			2d\sum_{k\ge1}\lambda_k\nonumber\\
			&=
			2d(1-\lambda_0)\nonumber\\
			&=
			2d(1-\|\Lambda\|_\infty),
		\end{align}
		which proves the statement. Again, in the step from the fourth to the fifth line, we used the fact that for a positive matrix $A$ we have $\|A\|_1=\tr A$.
\end{itemize}
\end{proof}

\section{Conclusion}\label{sec:conclusion}
In summary, we have provided a framework to unify the random dynamical decoupling with the quantum Zeno effect through the average of dynamical decoupling evolution and a phenomenon we call equitability of system and bath. By this approach, we infer explicit error bounds for the random dynamical decoupling only from an error bound on the quantum Zeno limit.

Our framework makes use of the fact that the average random dynamical decoupling evolution is a manifestation of the quantum Zeno dynamics. In the Zeno limit, we obtain a unitary evolution on the bath through this average protocol. In turn, with a high probability, trajectory random dynamical decoupling evolutions converge to a unitary bath evolution in the decoupling limit as well. This observation can be transferred to the system by studying the Schmidt decomposition of the Choi-Jamio\l{}kowsi state of the full evolution: As it is a pure state for each trajectory evolution, its reduced states for the two systems have the same spectrum. Our discussion shows that the system and the bath can get assigned the same role in the context of quantum control and therefore can be treated equitable: There is no dedicated subsystem for control, i.e., if one subsystem evolves unitarily, the other subsystem will do so as well.

\section*{Acknowledgments}
AH would like to thank Mattias Johnsson for useful advice on numerics. In addition, AH acknowledges support by Leibniz Universit\"at Hannover and in particular by Tobias Osborne, who provided the mandatory infrastructure and was always available for helpful discussions. AH was supported by the Sydney Quantum Academy. DB acknowledges funding by the Australian
Research Council (project numbers FT190100106, DP210101367, CE170100009).
This work was supported in part by the Top Global University Project from the Ministry of Education, Culture, Sports, Science and Technology (MEXT), Japan. KY was supported by the Grants-in-Aid for Scientific Research (C) (No.~18K03470) and for Fostering Joint International Research (B) (No.~18KK0073) both from the Japan Society for the Promotion of Science (JSPS).

\appendix
\section{Basic Lemmas}\label{appendix:lemmas}
In this appendix, we prove five basic lemmas. The first one (Lemma~\ref{lemma:telescope}) is a standard telescope sum identity, which we use to prove the bound on the quantum Zeno limit, i.e.~in the proof of Theorem~\ref{thm:convergence_average}. The second one (Lemma~\ref{lemma:tr1_norm}) and the third one (Lemma~\ref{lemma:reduced_map}) are used in the discussion about the average dynamical decoupling. In particular, we need them in order to prove Proposition~\ref{prop:Av2D}. They are used to convert the norm of a reduced map to the the norm of the matrix representation of the full map. The fourth one (Lemma~\ref{lemma:pur_ev}) is an elementary statement, which relates the purity of a Choi-Jamio\l{}kowski state with its operator norm. The last lemma (Lemma~\ref{lemma:diamond_choi}) gives a norm equivalence between the diamond norm distance between two CPTP maps and the trace distance of their respective Choi-Jamio\l{}kowski states.

\begin{lemma}\label{lemma:telescope}
	For any
		$n\in\mathbb{N}$ and for any operators $A$ and $B$, we have
		\begin{equation}
			A^{n}-B^{n}=\sum_{k=0}^{n-1}A^{k}(A-B)B^{n-1-k}.
		\end{equation}
\end{lemma}
\begin{proof}
	This lemma can be proved by performing an index shift,
		\begin{align}
		\sum_{k=0}^{n-1}A^{k}(A-B)B^{n-1-k}&=\sum_{k=0}^{n-1}A^{k+1}B^{n-1-k}-\sum_{k=0}^{n-1}A^{k}B^{n-k}\nonumber\\
		&=\sum_{k=0}^{n-1}A^{k+1}B^{n-1-k}-\sum_{k=0}^{n-2}A^{k+1}B^{n-k-1}-B^n\nonumber\\
		&=A^n-B^n.
		\vphantom{\sum_{k=0}^{n-1}}
		\end{align}
\end{proof}

\begin{lemma}\label{lemma:tr1_norm}
Let $A$ be an operator acting on a bipartite Hilbert space $\mathscr{H}_1\otimes\mathscr{H}_2$, and $\ket{u}_1,\ket{v}_1\in\mathscr{H}_1$.
Then,
\begin{equation}
\|{}_1\bra{u}A\ket{v}_1\|_\infty
\le\|A\|_\infty\|\ket{u}_1\|\|\ket{v}_1\|,
\end{equation}
where $\|\bullet\|$ for vectors denotes the Euclidean norm in $\mathscr{H}_1$.
\end{lemma}
\begin{proof}
Note that
\begin{equation}
A_2={}_1\bra{u}A\ket{v}_1
\end{equation}
acts on $\mathscr{H}_2$.
Then, there exists a normalized $\ket{\psi_\mathrm{max}}_2\in\mathscr{H}_2$ such that
\begin{align}
\|A_2\|_\infty^2
&={}_2\bra{\psi_\mathrm{max}}A_2^\dag A_2\ket{\psi_\mathrm{max}}_2
\nonumber\\
&=\Bigl({}_1\bra{v}\otimes{}_2\bra{\psi_\mathrm{max}}\Bigr)\,A^\dag\,\Bigl(\ket{u}_1\bra{u}\otimes\mathbb{1}_2\Bigr)\,A\,\Bigl(\ket{v}_1\otimes\ket{\psi_\mathrm{max}}_2\Bigr)
\nonumber\\
&\le\|\ket{u}_1\|^2
\Bigl({}_1\bra{v}\otimes{}_2\bra{\psi_\mathrm{max}}\Bigr)\,A^\dag A\,\Bigl(\ket{v}_1\otimes\ket{\psi_\mathrm{max}}_2\Bigr)
\nonumber\\
&\le\|\ket{u}_1\|^2\|\ket{v}_1\|^2\|A\|_\infty^2.
\end{align}
\end{proof}

\begin{lemma}\label{lemma:reduced_map}
Let $\sigma_1$ be an operator acting on a Hilbert space $\mathscr{H}_1$ of dimension $d_1$ and let $A_2$ be an operator acting on a Hilbert space $\mathscr{H}_2$. Let $\mathcal{T}$ denote a linear map on operators acting on $\mathscr{H}_1\otimes\mathscr{H}_2$ (not necessarily CPTP). Then,
\begin{equation}
\sup_{\|A_2\|_2=1}\|{\tr_1}[\mathcal{T}(\sigma_1\otimes A_2)]\|_2
=\|
(\mathbb{1}_1|\hat{\mathcal{T}}|\sigma_1)\|_\infty
\le
\sqrt{d_1}\,\|\sigma_1\|_2\|\hat{\mathcal{T}}\|_\infty.
\end{equation}
\end{lemma}
\begin{proof}
The Frobenius norm of the image of an operator under a linear map can be related to the operator norm of the matrix representation of the map according to~\eqref{eq:22norm_infnorm}. The partial trace on the level of matrix representations of linear maps is discussed in Ref.~\cite[Sec.~VD]{Wood2015}. Combining these two facts gives
\begin{align}
\sup_{\|A_2\|_2=1}\|{\tr_1}[\mathcal{T}(\sigma_1\otimes A_2)]\|_2
&=
\|
(\mathbb{1}_1|\hat{\mathcal{T}}|\sigma_1)\|_\infty,
\intertext{which can be bounded by Lemma~\ref{lemma:tr1_norm} as} 
\|(\mathbb{1}_1|\hat{\mathcal{T}}|\sigma_1)\|_\infty
&\le
\sqrt{(\mathbb{1}_1|\mathbb{1}_1)}
\sqrt{(\sigma_1|\sigma_1)}
\|\hat{\mathcal{T}}\|_\infty
\nonumber\\
&=
\sqrt{d_1}\,\|\sigma_1\|_2\|\hat{\mathcal{T}}\|_\infty.
\end{align}
\end{proof}

\begin{lemma}\label{lemma:pur_ev}
Let $\Lambda$ be a density operator with spectral decomposition $\Lambda=\sum_k\lambda_k|v_k)(v_k|$ and purity $\pur=\sum_k\lambda_k^2$. Then,
\begin{equation}
\|\Lambda\|_\infty^2
\leq
\pur
\leq
\|\Lambda\|_\infty.
\end{equation}
\end{lemma}
\begin{proof}
	The statement follows from the definition of $\pur$ and $\infnorm{\Lambda}$ as well as the fact that $0\leq\lambda_k\leq 1$, $\forall k$, and $\sum_k\lambda_k=1$. Denoting by $\lambda_\mathrm{max}$ the largest eigenvalue among $\{\lambda_k\}$,
\begin{equation}
\|\Lambda\|_\infty^2
=\lambda_\mathrm{max}^2
\leq \sum_k \lambda_k^2
=\pur
\leq\sum_k\lambda_\mathrm{max}\lambda_k
=\|\Lambda\|_\infty.
\end{equation}
\end{proof}

\begin{lemma}\label{lemma:diamond_choi}
	Let $\mathcal{S}$ and $\mathcal{T}$ be linear maps acting on a $d$-dimensional quantum system. Let $\Lambda(\mathcal{S})$ and $\Lambda(\mathcal{T})$ be their Choi-Jamio\l{}kowski states. Then,
	\begin{equation}
		\frac{1}{d}\|
		\mathcal{S}-\mathcal{T}
		\|_\diamond
		\le
		\|
		\Lambda(\mathcal{S})-\Lambda(\mathcal{T})
		\|_1
		\le
		\|
		\mathcal{S}-\mathcal{T}
		\|_\diamond.
	\end{equation}
	\begin{proof}
	\mbox{}
	\begin{itemize}
	\item Upper bound:
		By linearity,
		\begin{align}
			\|
			\Lambda(\mathcal{S})-\Lambda(\mathcal{T})
			\|_1
			&=
			\left\|
			[(\mathcal{S}-\mathcal{T})\otimes\mathbb{I}]\!\left(\frac{1}{d}|\mathbb{1})(\mathbb{1}|
		\right)
			\right\|_1\nonumber\\
			&\le
			\sup_{\|A\|_1=1}
			\|
			[(\mathcal{S}-\mathcal{T})\otimes\mathbb{I}](A)
			\|_1\nonumber\\
			&=
			\|
			\mathcal{S}-\mathcal{T}
			\|_\diamond.
		\end{align}
		\item Lower bound:
		By Ref.~\cite[Proposition~3.38]{Watrous2018}, the supremum in the diamond norm distance is reached by a rank 1 operator, say $A=|X)(Y|$ with $\Fnorm{X}=1=\Fnorm{Y}$, i.e.,
		\begin{equation}
			\|
			\mathcal{S}-\mathcal{T}
			\|_\diamond
			=
			\left\|
			[
			(\mathcal{S}-\mathcal{T})\otimes\mathbb{I}
			]
			\bigl(|X)(Y|\bigr)
			\right\|_1.
		\end{equation}
		Using Roth's lemma~\eqref{Roth}, we can write $|X)=|\mathbb{1}X)=(\mathbb{1}\otimes X\transpose)|\mathbb{1})$ (and analogously for $Y$), and hence
		\begin{align}
			\|
			\mathcal{S}-\mathcal{T}
			\|_\diamond
			&=
			d\left\|
			[
			(\mathcal{S}-\mathcal{T})\otimes\mathbb{I}
			]
			\left((\mathbb{1}\otimes X\transpose)\frac{1}{d}|\mathbb{1})(\mathbb{1}|(\mathbb{1}\otimes Y\transpose)^{\dagger}\right)
			\right\|_1\nonumber\\
			&=
			d\|
			(\mathbb{1}\otimes X\transpose)
			[
			\Lambda(\mathcal{S})-\Lambda(\mathcal{T})
			]
			(\mathbb{1}\otimes Y\transpose)^{\dagger}
			\|_1.
\intertext{By H\"older's inequality,}
			\|
			\mathcal{S}-\mathcal{T}
			\|_\diamond
			&\le
			d\|
			\mathbb{1}\otimes X\transpose
			\|_\infty
			\|
			\mathbb{1}\otimes \overline{Y}
			\|_\infty
			\|
			\Lambda(\mathcal{S}) - \Lambda(\mathcal{T})
			\|_1\nonumber\\
			&\le
			d\|
			X\transpose
			\|_2
			\|
			\overline{Y}
			\|_2
			\|
			\Lambda(\mathcal{S}) - \Lambda(\mathcal{T})
			\|_1\nonumber\\
			&=
			d\|
			\Lambda(\mathcal{S}) - \Lambda(\mathcal{T})
			\|_1.
		\end{align}
	\end{itemize}
	\end{proof}
\end{lemma}

\section{Specification of the Models in the Numerical Simulations}\label{appendix:numerics}
This appendix is dedicated to the numerical simulations in the main text of this paper. For transparency reasons, we here specify all the models and the parameters used in this paper. Figures~\ref{fig:Zeno_full} and~\ref{fig:DD_full} are both based on the same model and parameters.

The model we study consists of one system qubit and one bath qubit, i.e.~$d_1=d_2=2$, and therefore $d=d_1d_2=4$. A Hamiltonian $H$ was chosen generically by the following procedure. First, we sampled a $d\times d$ matrix with complex entries, where both real and imaginary parts of each entry were randomly drawn between $-1$ and $1$. Then, we made it Hermitian by adding to this matrix its Hermitian conjugate. Afterwards, we made it traceless by subtracting $(\tr H/d)\mathbb{1}$. This Hamiltonian was then normalized, such that $\|H\|_\infty=1$. Last, we rounded each entry of the resulting matrix to two decimal places. This procedure gave rise to the matrix
\begin{equation}
	H=
	\left(
\begin{array}{cccc}
 -0.10 & -0.03-0.35\,\rmi & -0.22-0.36\,\rmi & 0.13  +0.21\,\rmi \\
 -0.03+0.35\,\rmi & -0.29 & 0.20  +0.27\,\rmi & -0.02-0.04\,\rmi \\
 -0.22+0.36\,\rmi & 0.20  -0.27\,\rmi & -0.33 & 0.30  +0.40\,\rmi \\
 0.13  -0.21\,\rmi & -0.02+0.04\,\rmi & 0.30  -0.40\,\rmi & 0.72 \\
\end{array}
	\right).
\end{equation}
This matrix can be written in the Pauli basis and its decomposition of the form~\eqref{Hamilton_decomp} is given by
\begin{align}
	H_1={}&{-0.120}\,X+0.200\,Y-0.195\,Z,\\
	H_2={}&0.135\,X-0.025\,Y-0.215\,Z,\\
	H_{12}={}&0.165\,X\otimes X + 0.030\,X\otimes Y -0.100\,X\otimes Z\nonumber\\
	&{}-0.240\,Y\otimes X + 0.035\,Y\otimes Y + 0.160\,Y\otimes Z\nonumber\\
	&{}-0.165\,Z\otimes X + 0.375\,Z\otimes Y + 0.310\,Z\otimes Z.
\end{align}
From this $H$, we obtained the matrix representation $\hat{\mathcal{H}}$ of the generator of the dynamics by
\begin{equation}
	\hat{\mathcal{H}}=H\otimes \mathbb{1} - \mathbb{1}\otimes H\transpose,
\end{equation}
where $\mathbb{1}$ is the $d\times d$ identity matrix. For all purposes, we chose $t=1/\|\mathcal{\hat{H}}\|_\infty$. Notice that our bounds only depend on the product $T=t\infnorm{\hat{\mathcal{H}}}$. Therefore, the norm of the Hamiltonian can effectively be rescaled by changing the free evolution time of the system. In our case, by construction we have $T=t\|\hat{\mathcal{H}}\|_\infty=1$.

Figure~\ref{fig:Zeno_full} compares the bound on the quantum Zeno limit from Theorem~\ref{thm:convergence_average} with a numerical simulation. The projection we use is the group average projection from the average dynamical decoupling introduced in~\eqref{DDprotocol}. It has a matrix representation
\begin{equation}
	\hat{\mathcal{D}}
	=
	\frac{1}{d_1}\vecket{\mathbb{1}_1}\vecbra{\mathbb{1}_1}\otimes\hat{\mathbb{I}}_{22'}
	=
	\left(
\begin{array}{cccccccccccccccc}
 \frac{1}{2} & 0 & 0 & 0 & 0 & 0 & 0 & 0 & 0 & 0 & \frac{1}{2} & 0 & 0 & 0 & 0 & 0 \\
 0 & \frac{1}{2} & 0 & 0 & 0 & 0 & 0 & 0 & 0 & 0 & 0 & \frac{1}{2} & 0 & 0 & 0 & 0 \\
 0 & 0 & 0 & 0 & 0 & 0 & 0 & 0 & 0 & 0 & 0 & 0 & 0 & 0 & 0 & 0 \\
 0 & 0 & 0 & 0 & 0 & 0 & 0 & 0 & 0 & 0 & 0 & 0 & 0 & 0 & 0 & 0 \\
 0 & 0 & 0 & 0 & \frac{1}{2} & 0 & 0 & 0 & 0 & 0 & 0 & 0 & 0 & 0 & \frac{1}{2} & 0 \\
 0 & 0 & 0 & 0 & 0 & \frac{1}{2} & 0 & 0 & 0 & 0 & 0 & 0 & 0 & 0 & 0 & \frac{1}{2} \\
 0 & 0 & 0 & 0 & 0 & 0 & 0 & 0 & 0 & 0 & 0 & 0 & 0 & 0 & 0 & 0 \\
 0 & 0 & 0 & 0 & 0 & 0 & 0 & 0 & 0 & 0 & 0 & 0 & 0 & 0 & 0 & 0 \\
 0 & 0 & 0 & 0 & 0 & 0 & 0 & 0 & 0 & 0 & 0 & 0 & 0 & 0 & 0 & 0 \\
 0 & 0 & 0 & 0 & 0 & 0 & 0 & 0 & 0 & 0 & 0 & 0 & 0 & 0 & 0 & 0 \\
 \frac{1}{2} & 0 & 0 & 0 & 0 & 0 & 0 & 0 & 0 & 0 & \frac{1}{2} & 0 & 0 & 0 & 0 & 0 \\
 0 & \frac{1}{2} & 0 & 0 & 0 & 0 & 0 & 0 & 0 & 0 & 0 & \frac{1}{2} & 0 & 0 & 0 & 0 \\
 0 & 0 & 0 & 0 & 0 & 0 & 0 & 0 & 0 & 0 & 0 & 0 & 0 & 0 & 0 & 0 \\
 0 & 0 & 0 & 0 & 0 & 0 & 0 & 0 & 0 & 0 & 0 & 0 & 0 & 0 & 0 & 0 \\
 0 & 0 & 0 & 0 & \frac{1}{2} & 0 & 0 & 0 & 0 & 0 & 0 & 0 & 0 & 0 & \frac{1}{2} & 0 \\
 0 & 0 & 0 & 0 & 0 & \frac{1}{2} & 0 & 0 & 0 & 0 & 0 & 0 & 0 & 0 & 0 & \frac{1}{2} \\
\end{array}
\right).
\end{equation}

In Fig.~\ref{fig:DD_full}, we compare the bounds from Sec.~\ref{sec:results} for random trajectory dynamical decoupling with a numerical simulation. The applied unitaries are chosen randomly from $\mathscr{V}=\{\mathbb{1},X,Y,Z\}$ with $X,Y,Z$ the Pauli matrices and $\mathbb{1}$ the $2\times 2$ identity matrix. Even though the set $\{\mathbb{1},X,Y,Z\}$ does not form a group it suffices to draw the decoupling operations from it to obtain the action of the group $\langle\mathbb{1},X,Y,Z\rangle=\{\pm\mathbb{1}, \pm X, \pm Y, \pm Z, \pm\rmi \mathbb{1}, \pm\rmi X, \pm\rmi Y, \pm\rmi Z\}$ generated by $\mathbb{1},X,Y,Z$. This is due to the fact that we only act via the adjoint representation.

For Fig.~\ref{fig:probability_full}, we used exactly the same model as before in order to show how the probability of ``bad'' random trajectories through $\mathscr{V}$ behaves. The numerical data for $\mathbb{E}\bigl[\pur\bigl(\Lambda_{1,\ket{0}\bra{0}}^{(j)}\bigr)\bigr]$ in Fig.~\ref{fig:typical_atypical_trajectory} coincide with those in Fig.~\ref{fig:Purity}. For the purity of $\Lambda_{1,\ket{0}\bra{0}}^\mathrm{typical}$ in Fig.~\ref{fig:typical_atypical_trajectory} we used the following typical random sequence of unitaries
\begin{align}
	j_\mathrm{typical}=
	\{&Z, Z, Y, Y, Y, Y, X, X, \mathbb{1}, \mathbb{1}, \mathbb{1}, Y, \mathbb{1}, \mathbb{1}, Y, Y, Y, Z, Z, X, Z, Y, Y, Y, Y,\nonumber\\
&Z, X, Y, Z, \mathbb{1}, X, Y, \mathbb{1}, Z, Y, X, Y, Y, Z, \mathbb{1}, Z, Z, X, Y, \mathbb{1}, Y, Z, X, \mathbb{1}, \mathbb{1},\nonumber\\
&Z, X, Y, Y, Y, X, Y, Z, Y, Y, X, Y, Y, Y, Y, \mathbb{1}, Z, Y, X, Y, Z, Z, X, \mathbb{1}, X,\nonumber\\
&\mathbb{1}, X, Y, Y, Z, X, Y, Z, X, \mathbb{1}, X, Z, Z, \mathbb{1}, Z, Y, X, X, \mathbb{1}, \mathbb{1}, Z, Y, Y, X, Y, \mathbb{1}\},
\end{align}
whereas for the purity of $\Lambda_{1,\ket{0}\bra{0}}^\mathrm{atypical}$ we used the following atypical random sequence of unitaries
\begin{align}
j_\mathrm{atypical}=
\{&\mathbb{1}, X, Y, \mathbb{1}, \mathbb{1}, \mathbb{1}, \mathbb{1}, \mathbb{1}, \mathbb{1}, \mathbb{1}, \mathbb{1}, \mathbb{1}, \mathbb{1}, \mathbb{1}, \mathbb{1}, \mathbb{1}, \mathbb{1}, \mathbb{1}, \mathbb{1}, \mathbb{1}, \mathbb{1}, \mathbb{1}, \mathbb{1}, \mathbb{1}, X,\nonumber\\
 &\mathbb{1}, X, \mathbb{1}, \mathbb{1}, \mathbb{1}, \mathbb{1}, \mathbb{1}, \mathbb{1}, \mathbb{1}, \mathbb{1}, \mathbb{1}, \mathbb{1}, \mathbb{1}, \mathbb{1}, \mathbb{1}, \mathbb{1}, \mathbb{1}, \mathbb{1}, \mathbb{1}, \mathbb{1}, \mathbb{1}, \mathbb{1}, \mathbb{1}, \mathbb{1}, \mathbb{1},\nonumber\\
 &\mathbb{1}, Z, Z, \mathbb{1}, \mathbb{1}, \mathbb{1}, \mathbb{1}, \mathbb{1}, \mathbb{1}, \mathbb{1}, Z, \mathbb{1}, \mathbb{1}, \mathbb{1}, \mathbb{1}, \mathbb{1}, \mathbb{1}, \mathbb{1}, \mathbb{1}, \mathbb{1}, \mathbb{1}, \mathbb{1}, \mathbb{1}, \mathbb{1}, \mathbb{1},\nonumber\\
 &\mathbb{1}, \mathbb{1}, \mathbb{1}, \mathbb{1}, Z, \mathbb{1}, \mathbb{1}, \mathbb{1}, \mathbb{1}, \mathbb{1}, \mathbb{1}, \mathbb{1}, \mathbb{1}, \mathbb{1}, \mathbb{1}, \mathbb{1}, \mathbb{1}, \mathbb{1}, \mathbb{1}, \mathbb{1}, \mathbb{1}, \mathbb{1}, \mathbb{1}, \mathbb{1}, \mathbb{1}, \mathbb{1}\},
\end{align}
where the identity $\mathbb{1}$ appears much more frequently than others.

\bibliography{BibDDEfficiencyPaper.bib}
\end{document}